\documentclass{rsauthor}	

\usepackage{eucal}

\usepackage[sectionbib]{natbib}

\bibpunct[, ]{(}{)}{;}{a}{}{,}

\let\mycite\citep
\let\myciteasnoun\citet

\renewcommand\ge\geqslant
\renewcommand\geq\geqslant
\renewcommand\le\leqslant
\renewcommand\leq\leqslant
\renewcommand\propto\varpropto

\newcommand{\defeq}{:=}

\newcommand{\bbeta}{{\boldsymbol\beta}}
\newcommand{\bepsilon}{{\boldsymbol\epsilon}}
\newcommand{\btheta}{{\boldsymbol\theta}}
\newcommand{\bTheta}{{\boldsymbol\Theta}}

\DeclareMathOperator{\tr}{tr}
\DeclareMathOperator{\Var}{Var}

\makeatletter
\newenvironment{sizealignat}[2]{%
  \skip@=\baselineskip
  #1%
  \baselineskip=\skip@
  \alignat
  #2%
}{\endalignat \ignorespacesafterend}  
\makeatother

\newtheorem{mytheorem}{Theorem}
\newtheorem{mycorollary}[mytheorem]{Corollary}

\theoremstyle{definition} 
\newtheorem*{selectionpolicy*}{Selection Policy}
\newtheorem{mydefinition}[mytheorem]{Definition}

\newtheorem*{hypothesistest*}{Hypothesis Test}

\theoremstyle{remark}
\newtheorem*{remark*}{Remark}

\numberwithin{mytheorem}{section}

\jname{Proc.\ R. Soc.\ A}

\makeatletter
\def\titlepagefooter{\parbox{\textwidth}{%
\scriptsize \textit{\@jname}\ \@jpages;\ \@jdoi\@artnum\\%
\ifx\empty\@date\else\@date\vskip6pt\fi%
This journal is \copyright\ 2013 The Royal Society}}
\makeatother

\begin{document}

\title[Information criteria and normal regression]{Information criteria for deciding between\\ normal regression models}
\author[R. S. Maier]{\sc By Robert S. Maier${}^*$\footnote{${}^*$rsm@math.arizona.edu}}
\address{Departments of Mathematics and Physics, and Statistics Program\\ University of Arizona, Tucson, AZ 85721, USA}

\abstract{Regression models fitted to data can be assessed on their
goodness of fit, though models with many parameters should be
disfavored to prevent over-fitting.  Statisticians' tools for
this are little known to physical scientists.  These include the
Akaike Information Criterion (AIC), a penalized goodness-of-fit
statistic, and the AICc, a variant including a small-sample
correction.  They entered the physical sciences through being
used by astrophysicists to compare cosmological models; e.g.,
predictions of the distance--redshift relation.  The AICc is shown
to have been mis-applied, being applicable only if error
variances are unknown.  If error bars accompany the data, the AIC
should be used instead.  Erroneous applications of the AICc are
listed in an appendix.  It is also shown how the variability of
the AIC difference between models with a known error variance can
be estimated.  This yields a significance test that can
potentially replace the use of `Akaike weights' for deciding
between such models.  Additionally, the effects of model
mis-specification are examined.  For regression models fitted to
data sets without (rather than with) error bars, they are major:
the AICc may be shifted by an unknown amount.  The extent of this
in the fitting of physical models remains to be studied.}


\makeatletter
\markboth{\@shortauthor\hfil}{\hfil\@shorttitle}
\makeatother

\maketitle

\section{Introduction}
\label{sec:1}

\subsection{Background and overview}
\label{subsec:1a}

Physical scientists are familiar with the task of fitting a parametric
model such as a regression model to a data set, by using maximum likelihood
estimation (MLE) or other parameter estimation techniques.  But they are
less familiar with \emph{model selection}: deciding among two or more
models fitted to the same data in a way that for each model takes account
of both its goodness of~fit and its number of parameters.  To what extent
should one attempt to prevent over-fitting, i.e., `fitting the noise,' by
penalizing models with too many parameters?  This question of parsimony can
be viewed not only as a problem in data analysis, but as one in the
philosophy of science~\mycite{Zellner2001}.  The case when the models being
compared are incompatible, i.e., are non-nested in that they are not
related by parametric restrictions, is especially vexing.  So is the case
when they are mis-specified, i.e., do not agree with the `truth' (the
unknown and perhaps infinite-dimensional data-generating process),
regardless of what values for their parameters are chosen; so that fitting
errors of non-zero mean are present.

Techniques for model selection that penalize over-fitting have been applied
in the life sciences, social sciences and econometrics, and several
book-length expositions of these techniques by statisticians are
available~\mycite{Sakamoto86,McQuarrie98,Claeskens2008,Konishi2008}.
A~fruitful concept is the AIC (Akaike Information Criterion), a certain
penalized likelihood or goodness of~fit statistic (\citealp{Akaike73},
reprinted as \citealp{Akaike92}).  It is an estimate of the discrepancy, in
a sense related to MLE and information theory, between a fitted statistical
model and the unknown data-generating process; the latter being statistical
also, if measurement uncertainties are incorporated.  In simple cases the
AIC is effectively a penalized sum of squared prediction errors.  By
comparing AIC's one can compare models with different numbers of
parameters, including incompatible models.  But in the absence of a
systematic theory of the variability of the AIC statistic, using AIC's to
decide between fitted versions of parametric models
$\mathcal{M}^1,\mathcal{M}^2$ cannot be viewed as a procedure in classical
statistical inference, i.e., as a significance test.
No~$p$\nobreakdash-value, as in a frequentist assessment of the evidence
against a null hypothesis, is actually calculated.  Instead one simply
says, e.g., that if $\Delta^{12}\defeq\allowbreak{\rm AIC}^2-\nobreak {\rm
  AIC}^1$ is less than~$2.0$, the evidence that $\mathcal{M}^1$~is to be
preferred over~$\mathcal{M}^2$ is weak; and that if $\Delta^{12}$~is
greater than~$5.0$, it is strong.  The `Akaike weight' $\exp(-{\rm
  AIC}^l/2)$ is often viewed as an unnormalized probability (in~some sense)
that $\mathcal{M}^l$~is to be preferred~\mycite{Burnham2002}, but this
interpretation has not been universally accepted.

A few years ago, the AIC and related criteria (such as AICc, a variant
including a small correction) entered the physical sciences by being
introduced into astrophysics \mycite{Takeuchi2000,Liddle2007}.  They have
been used to compare cosmological models, such as regression models of the
distance--redshift relation that characterizes the expansion of the
Universe.  Unfortunately, in many papers a mistake in data analysis has
been made.  It can perhaps be attributed to a misreading of the expositions
of \myciteasnoun{Burnham2002} and~\myciteasnoun{Liddle2007}.  The mistake
is this.  A~data set may be accompanied by error bars (i.e., measurement
uncertainties), or not.  If the latter, regression model fitting will
involve the estimation of a `nuisance parameter,' namely the unknown
variance~$\sigma^2$ of the measurement errors.  The AICc, which was
designed to unbias completely the estimate of the Kullback--Leibler
information-theoretic discrepancy provided by the AIC, is appropriate
\emph{only} if $\sigma^2$~is unknown.  But data in the physical sciences
are typically accompanied by error bars.  When assessing statistical models
that incorporate known error bars, the AIC and not the AICc should be used.

To show this, we first place the AIC in a general framework that can be
used to derive many information-theoretic model-selection statistics.
(See~\S\,\ref{sec:2}.)  In~\S\,\ref{sec:3}\,$\ref{subsec:3a}$ we restrict
our focus to linear regression and MLE, and show that the applicability of
the AICc is limited as claimed.  In Appendix~B, papers from the
astrophysics literature that have erroneously applied the AICc are listed.

In~\S\,\ref{sec:3}\,$\ref{subsec:3b}$ we obtain a further result: under
reasonable conditions of mis-specification, using the correct statistic
(the AIC) to decide between normal regression models
$\mathcal{M}^1,\mathcal{M}^2$ that incorporate error bars can indeed be
viewed as a test of significance.  That~is, the decision can be made in a
classical way.  One can calculate a $p$\nobreakdash-value associated to the
null hypothesis that $\mathcal{M}^1,\mathcal{M}^2$ are equidistant in an
information-theoretic sense from the true but unknown
model~$\mathcal{M}^*$, as opposed to the alternative hypothesis that they
are not.  This is because the variability of $\Delta^{12}=\allowbreak{\rm
  AIC}^2-\nobreak {\rm AIC}^1$ can be estimated.  For data sets with error
bars, this can potentially render Akaike weights obsolete.  It is explained
how the estimation can be carried~out for mis-specified normal linear
models, and a hypothesis test based on the estimate is proposed.  The test
can be extended to non-linear models.

In decisions between statistical models $\mathcal{M}^1,\mathcal{M}^2$ that
incorporate known error bars, the validity of the uncorrected AIC is
unaffected if the models are mis-specified, as is shown
in~\S\,\ref{sec:3}\,$\ref{subsec:3a}$.  This result is unexpected, since
the usual derivation of the AIC statistic (and indeed of the AICc) requires
that there be nesting and no mis-specification; and its widespread
application to non-nested, potentially mis-specified models has in~fact
been somewhat heuristic.  In~\S\,\ref{sec:4} we show that if the data set
to which $\mathcal{M}^1,\mathcal{M}^2$ are fitted is \emph{not} accompanied
by error bars, problems with the AICc can indeed arise.  For a normal
linear model fitted to a data set without error bars, we discuss the
behavior of the AICc under mis-specification, and obtain an asymptotically
exact expression for the resulting shift.  If the extent of the
mis-specification is unknown, this shift may render the AICc of little
value.  This fact deserves to be better known.

Besides deriving the AIC and AIC corrections from first principles, we
briefly discuss the applicability to normal linear models of such variants
as ${\rm AIC}_\gamma$, which suppresses over-fitting to a greater extent
than does the~AIC\null.  Many additional variants have appeared in the
literature, such as the KIC (Kullback Information Criterion) and KICc
\mycite{Cavanaugh1999,Cavanaugh2004}, but they are beyond the scope of this
paper.  In the final section (\S\,\ref{sec:5}), we summarize our results.


\subsection{AIC basics}
\label{subsec:1b}

A regression model fitted to data can be linear or non-linear, according to
its parameter dependence.  The linear case is familiar
\mycite{Bevington2003, Draper98, Weisberg2005}.  Suppose the data set
comprises $y_1,\dots,y_n\in\mathbb{R}$; which could be, e.g., the values of
a response variable corresponding to $n$~distinct values of an explanatory
variable~$x$, chosen by an observer or an experimenter.  Suppose that
$y_1,\dots,y_n$ would depend linearly on parameters
$\beta_1,\dots,\beta_k\in\mathbb{R}$ in the absence of measurement errors
or other noise.  That is, ${\bf y}={\bf X}{\bbeta}$, where ${\bf
  y}=(y_i)_{i=1}^n$, ${\bbeta}=(\beta_j)_{j=1}^k$ are column vectors and
${\bf X}$~is an $n\times k$ design matrix.  (It~will be assumed throughout
that $n>k$ and that $\mathbf X$~is of full rank, i.e., of rank~$k$.)
A~statistical model~$\mathcal{M}$ of the data would then be
\begin{equation}
  y_i = \sum\nolimits_{j=1}^k X_{ij}\beta_j + \epsilon_i,
\end{equation}
where $\epsilon_1,\dots,\epsilon_n$ are residuals, i.e., errors.  In the
simplest case the residuals would be taken to be independent.  In a
(homoscedastic) normal model one would also take $\epsilon_i\sim
N(0,\sigma^2)$, i.e., take each~$\epsilon_i$ to be normally distributed
with mean zero and a common variance~$\sigma^2$.  The parameter~$\sigma^2$
may be known, or it may be an unknown nuisance parameter that needs to be
estimated (which is the case if error bars are not supplied).  Note that
from a data set $\bar{\bf y}=(\bar{y}_i)_{i=1}^n$ including error bars of
differing lengths, i.e., known but differing variances
$\sigma_1^2,\dots,\sigma_n^2$, one can obtain a data set ${\bf
  y}=({y}_i)_{i=1}^n$ with a known common variance~$\sigma^2$ by defining
${y}_i\defeq (\sigma/\sigma_i)\bar y_i$.

Whether or not $\sigma^2$ is known, an estimate
$\hat{\bbeta}=(\hat\beta_j)_{j=1}^k$ of~$\bbeta$ can be computed by MLE,
which reduces to ordinary least-squares for any normal linear model with
independent, identically distributed (i.i.d.)\ errors.  By a standard
calculation, $\hat{\bbeta}=({\bf X}^t{\bf X})^{-1}{\bf X}^t\,{\bf y}$.
Accompanying the observed data vector~${\bf y}$ there is then a predicted
data vector $\hat {\bf y} \defeq {\bf X}\hat{\bbeta} =\allowbreak {\bf
  P}{\bf y}$, where the `hat' matrix ${\bf P}={\bf X}({\bf X}^t{\bf
  X})^{-1}{\bf X}^t$ projects onto the column space of~${\bf X}$ (the
estimation space~$L\subset\mathbb{R}^n$).  The residual sum of squares
(RSS) for the fit is the sum of squared errors.  That is,
\begin{equation}
  \textrm{RSS} = ({\bf y} - \hat{\bf y})^t({\bf y} - \hat{\bf y})
= ({\bf y} - {\bf P}{\bf y})^t({\bf y} - {\bf P}{\bf y})
 = {\bf
    y}^t{\bf Q}{\bf y}, 
\end{equation}
where ${\bf Q}={\bf I}_n - {\bf P}$ is complementary to~${\bf P}$ and
projects onto the left null space of~${\bf X}$ (the error space
$L^\perp\subset\mathbb{R}^n$).  If $\sigma^2$~is known, so that the
parameter vector~$\btheta$ of~$\mathcal{M}$ is simply~$\bbeta$, the
standard definition of the AIC for the fitted model is
\begin{equation}
\label{eq:AIC}
  \textrm{AIC} = \textrm{RSS}/\sigma^2 + 2\,k.
\end{equation}
If alternatively $\sigma^2$ is unknown, so that
$\btheta=(\bbeta;\sigma^2)$, it is
\begin{equation}
\label{eq:AICsigma}
  \textrm{AIC} = n\ln\left(\hat\sigma^2\right) + 2\,(k+1) =
  n\ln\left(\textrm{RSS}/n\right) + 2\,(k+1),
\end{equation}
where $\hat\sigma^2 = \textrm{RSS}/n$ is the maximum likelihood estimate
of~$\sigma^2$.  

In both (\ref{eq:AIC}) and~(\ref{eq:AICsigma}) the first term equals up~to
an unimportant constant the statistic $-2\ln\mathcal{L}(\hat\btheta_N|{\bf
  y})$, where $\ln\mathcal{L}(\hat\btheta_N|{\bf y})$ is the log-likelihood
of the fitted model.  So the first term is a measure of goodness of fit.
In model selection a smaller AIC is preferred; hence the second term (which
will be seen to originate as an unbiasing term) penalizes~$\mathcal{M}$
according to its number of fitted parameters ($k$, resp.\ $k+1$).  It is
usually \emph{differences} of AIC's that are important, so any term not
involving~$k$ may be added to (\ref{eq:AIC}) and~(\ref{eq:AICsigma}).

The choice `2' in (\ref{eq:AIC}) and~(\ref{eq:AICsigma}) for the
coefficient of $k$ (resp.~$k+1$) is motivated by information theory, as
will be explained.  But applied statisticians have long been interested in
the effects on model selection of choosing a more general penalty
term~$\gamma k$, where $\gamma>0$ may differ from~2.  The resulting
modified AIC is denoted ${\rm AIC}_\gamma$ \mycite{McQuarrie98}.
\myciteasnoun{Bhansali77} considered the effects of varying~$\gamma$ on
order selection in autoregressive models, and showed empirically that it
may be useful for $\gamma$ to range, say, between 2 and~6.  The
abovementioned KIC like the AIC has an information-theoretic justification,
and in the $n\to\infty$ limit turns~out to be equivalent to~${\rm AIC}_3$.

Mention should also be made of the AICu \mycite{McQuarrie97}, which is a
heuristic modification of~(\ref{eq:AICsigma}) in which the ML
estimate~$\hat\sigma^2$ is replaced by the unbiased estimator
$s^2\defeq\textrm{RSS}/(n-\nobreak k)$ of~$\sigma^2$.  In the $n\to\infty$
limit, AICu is also equivalent to~${\rm AIC}_3$.  This can be seen by
working to leading order in~$1/n$ and using the asymptotic approximation $n
\ln[n/(n-\nobreak k)]\sim\allowbreak k +\nobreak O(1/n)$, $n\to\infty$.

The most familiar modified or corrected AIC, AICc, is a less drastic
modification of~(\ref{eq:AICsigma}), the modification being a major one
only for small~$n$.  Under the assumption of i.i.d.\ normal residuals, and
the traditional assumption of no mis-specification, it is given by
\begin{subequations}
\label{eq:AICc}
\begin{align}
\textrm{AICc} &= n\ln\left(\hat\sigma^2\right) + \frac{2(k+1)n}{n-k-2}\\
&\sim\textrm{AIC} + \frac{2\,(k+1)(k+2)}{n} + O(1/n^2),\qquad n\to\infty
\end{align}
\end{subequations}
\mycite{Sugiura78,Hurvich89,Cavanaugh97}.  Why an $O(1/n)$ correction term
should be added to~(\ref{eq:AICsigma}), but not to the
expression~(\ref{eq:AIC}) that applies if $\sigma^2$~is known, will be
explained.

\section{Minimum discrepancy estimation}
\label{sec:2}

\subsection{A general framework}
\label{subsec:2a}

In this section the AIC, a penalized goodness-of-fit statistic, is placed
in a model-selection framework that goes well beyond the use of MLE in
regression.  The AIC for a candidate model can be viewed as an unbiased
estimator of its discrepancy, in a certain sense, from the true model.
This will eventually lead to the introduction of a null hypothesis that two
candidate models are equally discrepant, and to systematic results on AIC
corrections.  But the theme of this section is the existence of
alternatives to the AIC, which have not yet been applied in the physical
sciences.  It is hoped that interest in this area will be stimulated.
A~framework resembling the one used here was first developed by
\myciteasnoun{Linhart86}.

Suppose one has a parametric statistical model~$\mathcal{M}_\btheta$ that
will be used for approximation or fitting purposes, with
$\btheta\in{\bTheta}$ (a~parameter space); and a true, underlying
statistical model~$\mathcal{M}^*$ of the data-generating process, which is
not known explicitly.  If each generated datum is an element of a set~$S$,
both $\mathcal{M}_\btheta$ and~${\mathcal{M}}^*$ will be probability
distributions on~$S$.  (The choice $S=\mathbb{R}^n$ is appropriate for
regression, as in~\S\,\ref{sec:1}\,$\ref{subsec:1b}$.)  Their respective
probability density functions (PDF's) will be denoted $f_\btheta({\bf y})$
and~${g}({\bf y})$.  In general it will not be assumed that
${g}=f_{\btheta_*}$ for any~$\btheta\in{\bTheta}$, i.e., mis-specification
will be allowed.  To any random sample ${\mathfrak{y}}_N$ of size~$N$ from
the true distribution~${g}$, comprising ${\bf y}^{(1)},\dots,{\bf
  y}^{(N)}\in S$, there corresponds an empirical distribution
$g_N=g_{N,{\mathfrak{y}}_N}$ on~$S$.  It is defined by
  \begin{equation}
    \label{eq:deltas}
    g_{N,{\mathfrak{y}}_N}(\cdot) = N^{-1}\sum\nolimits_{i=1}^N \delta({\cdot-{\bf y}^{(i)}}),
  \end{equation}
where $\delta(\cdot)$~is the Dirac delta function, if $S$~is a Euclidean
space such as~$\mathbb{R}^n$.  If alternatively $S$~is a discrete space,
then instead of a PDF there will be a probability mass function (PMF),
defined using a Kronecker delta rather than a delta function.  The
restriction $N=1$, meaning that there is only a single replication,
i.e.\ only one observation of ${\bf y}\in S$, was implicitly made
in~\S\,\ref{sec:1}\,$\ref{subsec:1b}$, where ${\bf y}$~was a random vector
in~$\mathbb{R}^n$; but here it will be relaxed.

The definition of MLE, which is an almost universally applicable but hardly
unique fitting scheme, is familiar.  From the sample~${\mathfrak{y}}_N$ one
constructs the likelihood function ${\cal L}(\btheta|{\mathfrak{y}}_N) =
\prod_{i=1}^N f_\btheta({\bf y}^{(i)})$, and computes a parameter estimate
$\hat\btheta_N = \hat\btheta_N({\mathfrak{y}}_N)$ by maximizing the
likelihood, or equivalently by minimizing the negative log-likelihood $-\ln
{\cal L}(\btheta|{\mathfrak{y}}_N)$.  The best fit to the data is then the
model~$\mathcal{M}_{\hat\btheta_N}$, with PDF~$f_{\hat\btheta_N}$.

This scheme generalizes to \emph{minimum discrepancy estimation}, which is
itself a generalization of minimum \emph{distance} estimation.  In an
abstract description one starts with~${\mathfrak{y}}_N$ or equivalently an
$N$\nobreakdash-point empirical distribution~$g_{N}$ on~$S$, and
computes~$\hat\btheta_N$ by minimizing $d(g_N;f_\btheta)$ over
$\btheta\in{\bTheta}$.  That is,
$\hat\btheta_N=\allowbreak\arg\min_{\btheta\in\bTheta} d(g_N;f_\btheta)$.
Here $d(g;f)$ signifies some real-valued measure of the discrepancy between
the PDF's $g,f$, which quantifies how difficult it is to discriminate
between them.  The case when $d$~satisfies the axioms for a metric can be
especially nice \mycite{Donoho88,Trosset95}, but this will not be assumed.
Thus $d$~may be asymmetric, i.e., directed, and may not satisfy the
triangle inequality.  Also, it may not satisfy $d(g;f)\geq0$.  But it is
useful to require that $d(g;f)\ge d(g;g)$, with equality holding only if
$g=f$.  Then the normalized discrepancy $D(g;f)\defeq\allowbreak d(g;f)
-\nobreak d(g;g)\ge0$ will satisfy $D(g;f)=0$ only if $g=f$.  As will be
seen, it is sometimes possible for an unnormalized discrepancy~$d$ to be
defined on a larger class of PDF's than is the case for a normalized one.

In minimum discrepancy estimation one must distinguish between the
model~$\mathcal{M}_\btheta$ fitted to an empirical distribution~$g_N$
generated by~$\mathcal{M}^*$, which has PDF $f_{\hat\btheta_N}$, and the
best approximating model, the PDF of which is some~$f_{\btheta_*}$.  Here
$\btheta_*=\allowbreak\arg\min_{\btheta\in\bTheta} d({g};f_\btheta)$ may
differ from~$\hat\btheta_N$.  The value~$\btheta_*$ is called the
`pseudo-true' value of~$\btheta$.  (If there is no mis-specification, i.e.,
${g}=f_{\btheta_*}$ for some~${\btheta_*}$, it is the true value.)  The
\emph{discrepancy due to approximation} (AD) is $d({g};f_{\btheta_*})$.  In
the absence of mis-specification this would equal the constant
$d({g};{g})$, and what would be more important would be the
\emph{discrepancy due to estimation} (ED),
i.e.\ $d(f_{\btheta_*};f_{\hat\btheta_N})$.  The \emph{overall discrepancy}
(OD), of the fitted model from the truth, is the quantity
$d({g};f_{\hat\btheta_N})$.  If there is no mis-specification, OD reduces
to~ED\null.

None of these three discrepancies can be calculated if the true model is
unknown, though they can be estimated from the data~${\mathfrak{y}}_N$,
meaning from~$g_N$.  (It is $d(g_N;f_{\hat\btheta_N})$, the \emph{fitted
  discrepancy} (FD) of the model from the data, that can be calculated from
the data.)  The following policy is an abstraction of Akaike's.

\begin{selectionpolicy*}
  Given a discrepancy functional~$d$, data~${\mathfrak{y}}_N$ generated by
  an unknown true model~${\mathcal{M}}^*$ with PDF~${g}$, and candidate
  models $\mathcal{M}^1_{\btheta^1},\mathcal{M}^2_{\btheta^2}$ with
  parametric PDF's $f^1_{\btheta^1}, f^2_{\btheta^2}$ and parameter spaces
  $\bTheta^1,\bTheta^2$, one should ideally assess the goodness of fit of
  each model on the basis of its \emph{expected overall discrepancy}
  from~${\mathcal{M}}^*$.  That is, if when fitted to
  data~${\mathfrak{y}}_N$ or equivalently to the empirical distribution
  $g_N=g_{N,{\mathfrak{y}}_N}$, model $\mathcal{M}^l_{\btheta^l}$ would
  have parameter $\hat\btheta_N^l=\hat\btheta_N^l({\mathfrak{y}}_N)$, one
  should select the model with the minimum value of
  \begin{equation}
    \label{eq:eod1}
    d'({g}, f^l_.) \defeq E_{{\mathfrak{y}}_N}\textrm{OD}_d^l({{\mathfrak{y}}_N})
    = E_{{\mathfrak{y}}_N}\,d\bigl({g},f_{\hat\btheta^l_N({\mathfrak{y}}_N)}^l\bigr).
  \end{equation}
  The expectation (i.e., averaging) is computed over
  data~${\mathfrak{y}}_N$ generated by the true model, meaning over the
  PDF~${g}$.  As an alternative, the double expectation
  \begin{equation}
    \label{eq:eod2}
    d''({g}, f^l_.) \defeq
    E_{{\mathfrak{y}}_N}E_{{\mathfrak{y}}'_N}\,d\bigl(g_{N,{\mathfrak{y}}'_N},f_{\hat\btheta^l_N({\mathfrak{y}}_N)}^l\bigr),
  \end{equation}
where data ${\mathfrak{y}}_N,{\mathfrak{y}}'_N$ are generated independently by~${g}$, may
be employed.
\end{selectionpolicy*}

The expected OD, $d'({g}, f^l_.)$, is a penalized version of the AD
$d({g};f_{\btheta^l_*})$, where $\btheta^l_*$~is the pseudo-true value of
the parameter $\btheta^l\in\bTheta^l$.  It~is at~least as large as the~AD,
and because the number of ways in which the data-dependent fitted PDF
$f_{\hat\btheta^l_N({\mathfrak{y}}_N)}^l$ \!can deviate from the pseudo-true
PDF $f^l_{\btheta^l_*}$ is the dimensionality of the parameter
space~$\bTheta^l$, one expects that relying on the expected OD as a measure
of closeness of $\mathcal{M}^l_{\hat\btheta^l_N}$ to~$\mathcal{M}^*$ will
disfavor over-fitting in an AIC-like way.

It~must be stressed that neither $d'({g}, f^l_.)$ nor $d''({g}, f^l_.)$ can
be calculated directly, though they can be estimated (with bias) by the
fitted discrepancy of the $i$'th model from the data,
$\textrm{FD}_d^l({\mathfrak{y}}_N)=d(g_{N,{\mathfrak{y}}_N};f_{\hat\btheta^l_N({{\mathfrak{y}}_N})})$.
(This is an RSS-like quantity.)  The selection policy can therefore be
implemented by choosing the model with the minimum value of a certain
function MSC, defined as follows to be an unbiased estimator of the
expected overall discrepancy.  (It will specialize to the AIC\null.)

\begin{mydefinition}
\label{def:1}
The model-selection criterion function based on a discrepancy
functional~$d$, denoted $\textrm{MSC}_d$ or simply MSC, is defined so that
the value $\textrm{MSC}^l/2$ for the $l$'th model equals its fitted
discrepancy
$\textrm{FD}_d^l({\mathfrak{y}}_N)=d(g_{N,{\mathfrak{y}}_N};f_{\hat\btheta^l_N({\mathfrak{y}}_N)})$,
plus an unbiasing term~$B^l$ equal to
$E_{{\mathfrak{y}}_N}\textrm{OD}_d^l({\mathfrak{y}}_N)$ minus
$E_{{\mathfrak{y}}_N}\textrm{FD}_d^l({\mathfrak{y}}_N)$.  That is,
\begin{align}
  B^l &= d'(g;f_.^l) - E_{{\mathfrak{y}}_N}\,d\bigl(g_{N,{\mathfrak{y}}_N};f_{\hat\btheta^l_N({\mathfrak{y}}_N)}\bigr)\nonumber\\
    &=   E_{{\mathfrak{y}}_N}\,d\bigl({g};f_{\hat\btheta^l_N({\mathfrak{y}}_N)}^l\bigr)
    - E_{{\mathfrak{y}}_N}\,d\bigl(g_{N,{\mathfrak{y}}_N};f_{\hat\btheta^l_N({\mathfrak{y}}_N)}\bigr),
\end{align}
so that in expectation, $\textrm{MSC}^l/2$ equals the expected overall
discrepancy, on which it is posited that model selection should be based.
\end{mydefinition}

\begin{remark*}
  Distinct discrepancies could be used for (i)~performing the fitting and
  computing the FD, and (ii)~computing the expected OD\null.  This would
  generalize the selection policy (and turns~out to be needed in the
  definition of the KIC, as will be discussed elsewhere).  For practical
  reasons, one could also perform the fitting using MLE and compute the FD
  using a discrepancy not related to MLE \mycite[\S\,4.4]{Linhart86}.  But
  this seems less theoretically justifiable.
\end{remark*}

The alternative $d''$ to~$d'$ was mentioned for two reasons.  First, for
MLE it is identical to~$d'$, as will be seen.  Second, using~$d''$ makes
selection in~effect a procedure of \emph{cross validation}, in which a
fitted model is assessed according to its empirically expected prediction
errors \mycite{Stone78}.  The definition~(\ref{eq:eod2}) of~$d''$ involves
two hypothetical sets of data: one (${\mathfrak{y}}_N$) used to estimate
the parameter~$\btheta^l$, and one (${\mathfrak{y}}'_N$) used to assess the
fit of the resulting model.  The equivalence between choosing models by
cross-validation and Akaike's technique of penalizing models by their
number of parameters has long been recognized \mycite{Stone77, Kuha2004}.

\subsection{Discrepancies, information theory and MLE}
\label{subsec:2b}

The selection policy of \S\,\ref{sec:2}\,$\ref{subsec:2a}$ can be applied
very generally, to both discrete and continuous models.  Given a
discrepancy functional~$d$ and a sample $\mathfrak{y}_N$ comprising ${\bf
  y}^{(1)}, \dots, {\bf y}^{(N)}\in S$, which yields an $N$-point empirical
distribution $g_N=g_{N,{\mathfrak{y}}_N}$ on~$S$, one would naively decide
whether to select parametric model~$\mathcal{M}^l_{\btheta^l}$ with PDF
$f^l_{\btheta^l}$ on~$S$ on the basis of its fitted discrepancy
$d(g_{N,{\mathfrak{y}}_N};f_{\hat\btheta^l_N({{\mathfrak{y}}_N})})$ from
the sample.  But the selection policy modifies this RSS-like quantity by
adding an unbiasing term, giving an unbiased estimate $\textrm{MSC}/2$ of
the expected overall discrepancy.  Selection based on the latter disfavors
over-fitting, as desired.

This is the natural setting for the AIC and AICc statistics.  But for
regression models, in which $S=\mathbb{R}^n$ (and most often $N=1$), the
choice of a discrepancy functional closely tied to MLE and in~fact to
information theory must first be justified.  Only for one special
functional does the fitted discrepancy turn~out to be the negative
log-likelihood.

For discrete rather than continuous models, with $S$ a discrete set such as
$\{1,\dots,m\}$ or $\{1,2,3,\dots\}$, a wide variety of discrepancy
functionals $d(g;f)$ have been used in statistics and elsewhere.  They
include the Pearson $X^2$ statistic $\sum\nolimits_{n\in
  S}\left[g(n)-\nobreak f(n)\right]^2/f(n)$, used when $S$~is the set of
cells in a contingency table to compare an empirical and a theoretical
distribution.  The Neyman~$X^2$ is similar.  There are also many
discrepancies rooted in information theory, such as the (discrete)
Kullback--Leibler (KL) divergence, often called the informational
divergence \mycite{Csiszar2011}.  Its normalized form is
\begin{equation}
  D_{\textrm{KL}}(g;f) \defeq \sum\nolimits_{n\in S} g(n)\ln \frac{g(n)}{f(n)} \,\ge\,0
\end{equation}
and its denormalized form is
\begin{equation}
\label{eq:dKLcontin}  
  d_{\textrm{KL}}(g;f) \defeq -\sum\nolimits_{n\in S} g(n)\ln f(n) \,\ge\, d_{\textrm{KL}}(g;g).
\end{equation}
They are related by $D_\textrm{KL}(g;f)=\allowbreak d_\textrm{KL}(g;f)
-\nobreak d_\textrm{KL}(g;g)$.  The subtracted quantity
$d_\textrm{KL}(g;g)$ is the (Shannon) entropy of the distribution~$g$, and
in statistical mechanics $D_{\textrm{KL}}(g;f)$ would therefore be called a
relative entropy.  The denormalized $d_{\textrm{KL}}(g;f)$ is sometimes
called an `inaccuracy.'

Discrepancies in information theory are often of the
`$\varphi$\nobreakdash-divergence' form
$D_\varphi(g;f)\defeq\allowbreak\sum\nolimits_{n\in S} g(n)\,
\varphi\left(f(n)/g(n)\right)$, for $\varphi$ some convex function
\mycite{Pardo2006}.  For instance, if $\varphi(u)$ equals $-\ln u$ then
$D_\varphi=D_{\textrm{KL}}$.  If $\varphi(u)\propto 1-\nobreak
u^{(1+\alpha)/2}$ then $D_\varphi$~becomes the so-called
$\alpha$\nobreakdash-divergence~$D^{(\alpha)}$, which reduces
to~$D_{\textrm{KL}}$ in a scaled $\alpha\to-1$ limit.  This generalized
discrepancy arises in the geometry of statistical inference
\mycite{Amari2000}, and there is a corresponding denormalized
\begin{equation}
d^{(\alpha)}(g;f)\propto\allowbreak \sum\nolimits_{n\in S}
g(n)^{(1-\alpha)/2}\left[1-\nobreak f(n)^{(1+\alpha)/2}\right],
\end{equation}
a scaled $\alpha\to-1$ limit of which equals $d_{\textrm{KL}}(g;f)$.  The
self-divergence $d^{(\alpha)}(g;g)$ of~$g$ is called in physics the Tsallis
entropy of~$g$ \mycite{Tsallis88}.

Many other discrepancies have been investigated.  (See
\myciteasnoun[Chap.~7]{Kapur89} and \myciteasnoun[Chap.~11]{Basu2011}.)
But the KL divergence is perhaps the most important, because of its
connection to~MLE\null.  From a sample ${\mathfrak{y}}_N$ comprising ${\bf
  y}^{(1)},\dots,{\bf y}^{(N)}\in S$, where $S$~is a discrete set,
obtaining a best-fit PMF~$f_{\hat\btheta_N}$ by maximizing over~$\btheta$
the likelihood $\mathcal{L}_f(\btheta|{\mathfrak{y}}_N)=\allowbreak
f({\mathfrak{y}}_N|\btheta)$ of a candidate PMF~$f_\btheta$ on~$S$ is
equivalent to minimizing $D_{\textrm{KL}}(g_N;f_{\btheta})$ or
$d_{\textrm{KL}}(g_N;f_{\btheta})$ over~$\btheta$, where
$g_N=g_{N,{\mathfrak{y}}_N}$~is the $N$\nobreakdash-point empirical
distribution defined by the data.  This is because
$d_{\textrm{KL}}(g_N;f_{\btheta})$ equals $-\ln
\mathcal{L}_f(\btheta|{\mathfrak{y}}_N)$, the negative log-likelihood, as
follows from~(\ref{eq:dKLcontin}).  In particular,
$d_{\textrm{KL}}(g_{N};f_{\hat\btheta_N})$ equals $-\ln
\mathcal{L}_f(\hat\btheta_N|{\mathfrak{y}}_N)$.  MLE~can thus be
interpreted as a minimum discrepancy estimation.

Now consider the continuous case, when the data lie in a space~$S$ that is
Euclidean, such as the choice $S=\mathbb{R}^n$ arising in regression.  The
selection policy of~\S\,\ref{sec:2}\,$\ref{subsec:2a}$ requires the
computation of (an unbiased version~of) some fitted discrepancy
$d(g_N;f_{\hat\btheta_N})$.  The empirical distribution
$g_N=g_{N,{\mathfrak{y}}_N}$ is a linear combination of $N$~delta functions
computed from the data~${\mathfrak{y}}_N$, as in~(\ref{eq:deltas}), and
$\hat\btheta_N=\hat\btheta_N({\mathfrak{y}}_N)$~is computed by minimizing
$d(g_N;f_{\btheta})$ over~$\btheta$, where $f_{\btheta}$~is a candidate PDF
on~$S$.  For the selection policy to be implemented as stated, the
discrepancy functional $d(g;f)$ being employed must allow its first
argument to be a PDF~$g_N$ that is `atomic,' in the sense that it is a
combination of delta functions.  This is a stringent requirement
\mycite[\S\,10.9]{Liese87}.

When $S=\mathbb{R}$, many statistical discrepancies $d(g;f)$ have been
employed; e.g., in the robust estimation of location and scale parameters
of distributions on~$S$ \mycite{Sahler68,Parr80}.  Most are computed from
the cumulative distributions (CDF's) $G,F$ corresponding to~$g,f$, so they
are well-defined even if one or the other is an empirical distribution.
Also, most are metrics, or legitimate distances; in~fact minimum
discrepancy estimation grew out of minimum distance estimation, which has a
long history \mycite{Parr81}.  Examples include the Kolmogorov--Smirnov and
Cram\'er--von~Mises discrepancies, which are widely used as goodness-of-fit
statistics.  But their generalizations to the multivariate case, when
$S=\mathbb{R}^n$ with $n>1$, are not straightforward at~all.

When $S=\mathbb{R}^n$ with $n$ arbitrary, the most widely used discrepancy
is the (continuous) KL divergence, with its close ties to MLE\null.  Its
normalized form is
\begin{equation}
\label{eq:DKL}
  D_{\textrm{KL}}(g;f) \defeq \int_{\mathbb{R}^n} g({\bf y})\ln \frac{g({\bf y})}{f({\bf
      y})}\,\,{\rm d}^n{\bf y} \,\ge\,0
\end{equation}
and its denormalized form is
\begin{equation}
\label{eq:dKL}
  d_{\textrm{KL}}(g;f) \defeq -\int_{\mathbb{R}^n} g({\bf y})\ln f({\bf y})\,\,{\rm d}^n{\bf y} \,\ge\, d_{\textrm{KL}}(g;g).
\end{equation}
They are related by $D_\textrm{KL}(g;f)=\allowbreak d_\textrm{KL}(g;f)
-\nobreak d_\textrm{KL}(g;g)$, as in the discrete case.  A key observation
is that the integral in~(\ref{eq:dKL}) is well-defined even if $g$ is a
purely atomic function of~${\bf y}$, such as an empirical PDF~$g_N$\null.
But the integral in~(\ref{eq:DKL}) is not, as one cannot take the logarithm
of a sum of delta functions.  The distinction can be viewed as arising from
the entropy $d_\textrm{KL}(g;g)$ not being defined when $g=g_N$: any
empirical distribution has undefined entropy.  The good behavior of the
integral in~(\ref{eq:dKL}) justifies the use of~$d_\textrm{KL}$ in model
selection, to the exclusion of~$D_\textrm{KL}$.

For any sample size $N\ge1$, averaging over data~${\mathfrak{y}}'_N$
sampled from the PDF~$g$ yields $g$~itself; which is to say,
$E_{{\mathfrak{y}}'_N}\,g_{N,{\mathfrak{y}}'_N}=g$.  By the linearity
in~$g$ of the integral in~(\ref{eq:dKL}), it follows that
$E_{{\mathfrak{y}}'_N}\,d_\textrm{KL}(g_{N,{\mathfrak{y}}'_N};f)$ equals
$E\,d_\textrm{KL}(g;f)$.  This confirms a claim made
in~\S\,\ref{sec:2}\,$\ref{subsec:2a}$: if $d_\textrm{KL}$ is used as the
discrepancy~$d$, the two expected overall discrepancies $d',d''$ defined in
(\ref{eq:eod1}),(\ref{eq:eod2}) are \emph{equal}, and give rise to
identical model-selection policies.  For non-KL discrepancies, this may
not hold.

In the continuous case as in the discrete, the fitted discrepancy
$\textrm{FD}_{d_{\textrm{KL}}}({\mathfrak{y}}_N)$, i.e.,
$d_\textrm{KL}(g_{N,{\mathfrak{y}}_N};f_{\hat\btheta_N}({\mathfrak{y}}_N))$,
equals $-\ln\mathcal{L}_f(\hat\btheta_N|{{\mathfrak{y}}_N})$, the negative
of the fitted log-likelihood.  Using $d_\textrm{KL}$ as the discrepancy, as
in the following definition, specializes the MSC of Definition~\ref{def:1}
to what will be called a MSC of AIC type.

\begin{mydefinition}
\label{def:2}
A $d_\textrm{KL}$-based model-selection criterion $\textrm{MSC}$, of AIC
type, is defined thus: the value $\textrm{MSC}^l/2$ for the $l$'th model,
calculated after fitting, equals its negative log-likelihood
$-\ln\mathcal{L}_f(\hat\btheta^l_N|{{\mathfrak{y}}_N})$, plus an unbiasing
correction~$B^l$ that equals
$E_{{\mathfrak{y}}_N}\textrm{OD}_{d_\textrm{KL}}^l({\mathfrak{y}}_N)$ minus
$E_{{\mathfrak{y}}_N}\textrm{FD}_{d_\textrm{KL}}^l({\mathfrak{y}}_N)$.
That is,
\begin{displaymath}
  B^l  = 
E_{{\mathfrak{y}}_N}\,d_\textrm{KL}\bigl({g};f_{\hat\btheta^l_N({\mathfrak{y}}_N)}^l\bigr)
- E_{{\mathfrak{y}}_N}\left[
-\ln\mathcal{L}_f(\hat\btheta_N^l|{{\mathfrak{y}}_N})
\right],
\end{displaymath}
so that in expectation, $\textrm{MSC}^l/2$ equals the $d_\textrm{KL}$-based
expected overall discrepancy, on which it is posited that model selection
should be based.
\end{mydefinition}

In the following, only MSCs of AIC type will be employed.  The choice of
the KL divergence as the discrepancy used in model fitting has a clear
justification: it allows the selection policy
of~\S\,\ref{sec:2}\,$\ref{subsec:2a}$ to be applied as stated.

But it should be noted that there are alternatives that merit examination.
By adapting the selection policy it may be possible to employ quite
different discrepancies, such as the abovementioned
$\alpha$\nobreakdash-divergence.  This requires a brief explanation.  The
continuous, denormalized version of the $\alpha$\nobreakdash-divergence is
\begin{equation}
\label{eq:alpha}
  d^{(\alpha)}(g;f)\propto\allowbreak \int_{\mathbb{R}^n}
g({\bf y})^{(1-\alpha)/2}\left[1-\nobreak f({\bf y})^{(1+\alpha)/2}\right]\,\,
{\rm d}^n{\bf y},
\end{equation}
which is undefined if $g$~is an empirical distribution.  But the $\alpha=0$
case of this, $d^{(0)}(g;f)$, is equivalent to the Hellinger distance,
which has long been used in parametric estimation \mycite{Beran77}.  From
data ${\mathfrak{y}}_N$, or an empirical distribution
$g_N=g_{N,{\mathfrak{y}}_N}$ defined from~it as in~(\ref{eq:deltas}), a
fitted parametric model~$f_{\hat\btheta}$ can indeed be found by minimizing
the Hellinger distance.  The fitting, though, involves a preliminary step:
replacing the delta functions of~(\ref{eq:deltas}) by approximate deltas.
That is, one first engages in kernel density estimation, by
convolving~$g_N$ with some integral kernel.  The integral
in~(\ref{eq:alpha}) will be well-defined if the resulting `smoothed'~$g$ is
used.  Alternatively, the model PDF~$f$ as~well as the data PDF~$g_N$ can
be smoothed \mycite[\S\,3.3]{Basu97}.  However, the smoothing of~$g_N$ can
apparently be justified only in the large-sample ($N\to\infty$) limit.  In
what follows $N=1$, and the limit taken (if any) will be $n\to\infty$.  The
usefulness in this setting of a preliminary smoothing of the empirical
distribution remains to be explored.

\section{Selection with error bars}
\label{sec:3}

In this section we specialize to the case when the true
model~$\mathcal{M}^*$ and the candidate models fitted to a set of $n$~data
points are normal regression models, incorporating known error bars as
explained in~\S\,\ref{sec:1}\,$\ref{subsec:1b}$.  We investigate the
model-selection criterion function given in Definition~\ref{def:2}
(a~$d_\textrm{KL}$-based MSC of AIC type).

In~\S\,\ref{sec:3}\,$\ref{subsec:3a}$ it is shown that irrespective of~$n$
and the extent of mis-specification, the MSC for a candidate linear model
reduces to the standard AIC of~(\ref{eq:AIC}): a~sum of squared residuals
penalized by~$2k$, i.e., by twice the number of parameters.  Interestingly,
it is possible to derive the modified AIC known as the ${\rm AIC}_\gamma$,
in which the penalty $\gamma k$ replaces~$2k$, by slightly modifying the
selection policy of~\S\,\ref{sec:2}.  But even when $n$~is small, the
standard AIC is never extended by an $O(1/n)$ correction term; thus the use
of the AICc is not appropriate here.  This realization is new.  In
Appendix~B, papers from the astrophysics literature that have erroneously
applied the~AICc are listed.

In~\S\,\ref{sec:3}\,$\ref{subsec:3b}$ the variability of
$\Delta^{12}\defeq\allowbreak \textrm{AIC}^2 -\nobreak \textrm{AIC}^1$ is
determined, and an asymptotically valid hypothesis test for model selection
that is based on~$\Delta^{12}$ is proposed.  At~any specified significance
level~$\alpha$, the test either rejects or accepts the null hypothesis that
the fitted models
$\mathcal{M}^1_{\hat\btheta^1},\mathcal{M}^2_{\hat\btheta^2}$ are equally
close to~$\mathcal{M}^*$ in the `expected overall discrepancy' sense
of~\S\,\ref{sec:2}.  For mis-specified linear models incorporating error
bars, this approach to model selection can potentially replace the
rule-of-thumb use of Akaike weights.  The variability estimation and the
test of significance can be extended to the case of non-linear models, as
is sketched.

\subsection{Expected discrepancies}
\label{subsec:3a}

Consider a true model~$\mathcal{M}^*$ and a candidate linear
model~$\mathcal{M}$ that are both normal, as
in~\S\,\ref{sec:1}\,$\ref{subsec:1b}$.  They are ${\bf y}={\bf
  y}_0+\bepsilon_0$ and ${\bf y}={\bf X}{\bbeta}+\bepsilon$, where
$\bbeta=(\beta_j)_{j=1}^k$ is a column vector of parameters, and the error
vectors $\bepsilon_0,\bepsilon$ have mean zero and covariance matrices
$\sigma_0^2{\bf I}_n,\sigma^2{\bf I}_n$.  Thus $\bepsilon_0=\sigma_0 {\bf
  z}_0$ and $\bepsilon=\sigma {\bf z}$, where ${\bf z}_0$ and~${\bf z}$ are
vectors of independent standard normal variables.  It is not assumed that
${\bf y}_0$~is in the column space of the $n\times k$ design matrix~${\bf
  X}$, i.e., mis-specification is allowed.

In this section the variance~$\sigma^2$ is specified and not estimated, so
the full parameter vector~$\btheta$ of~$\mathcal{M}$ is simply~$\bbeta$.
If the true variance~$\sigma_0^2$ is known, as is the case when the data
are accompanied by error bars, then it is natural to choose
$\sigma^2=\sigma_0^2$.  But for the moment this will not be assumed:
statistical as well as deterministic mis-specification will be allowed.

There is assumed to be only one observation ($N=1$), so only one instance
of the random vector ${\bf y}\in S=\mathbb{R}^n$ is available as a datum.
Thus ${\bf y}$~will be written for~${\mathfrak{y}}$, and the subscript~$N$
dropped.  MLE is equivalent to choosing $\bbeta\in\mathbb{R}^k$ so as to
minimize the discrepancy (in the $d_\textrm{KL}$ sense) of the
1\nobreakdash-point atomic PDF $g_{\bf y}(\cdot)=\delta(\cdot - {\bf y})$
on~$\mathbb{R}^n$ from~$\mathcal{M}_\bbeta$.  Equivalently, MLE minimizes
the negative log-likelihood $-\ln\mathcal{L}(\bbeta|{\bf y})$
of~$\mathcal{M}_\bbeta$.  It yields a fitted model
$\mathcal{M}_{\hat\bbeta}$, where the estimated parameter vector
$\hat\bbeta=\hat\bbeta({\bf y})\in\mathbb{R}^k$ is given by
$\hat{\bbeta}=({\bf X}^t{\bf X})^{-1}{\bf X}^t\,{\bf y}$.  The $n\times n$
hat matrix~${\bf P}$ and its complement ${\bf Q} = {\bf I}_n-\nobreak {\bf
  P}$, which project onto the estimation and error subspaces
of~$\mathbb{R}^n$, i.e., the column and left null spaces of~${\bf X}$, are
defined as usual by ${\bf P}={\bf X}({\bf X}^t{\bf X})^{-1}{\bf X}^t$.  The
predicted data vector~$\hat{\bf y}$ is defined by $\hat{\bf y}={\bf P}{\bf
  y}$.  The RSS (sum of squared residuals) is $({\bf y}-\nobreak\hat{\bf
  y})^t({\bf y}-\nobreak\hat{\bf y})=\allowbreak{\bf y}^t{\bf Q}{\bf y}$.

The policy of~$\S\,\ref{sec:2}$ requires that to the extent that it can be
estimated, the expected overall discrepancy $E_{\bf y}\,\textrm{OD}({\bf
  y})$ should be used for model selection.  The AIC-type selection
criterion $\textrm{MSC}$ (see Definition~\ref{def:2}) has the property that
$\textrm{MSC}/2$ for~$\mathcal{M}$ is an unbiased estimator of $E_{\bf
  y}\,\textrm{OD}({\bf y})$.  It is defined by
\begin{equation}
\label{eq:unbiasing}
  \textrm{MSC/2} = \textrm{FD}({\bf y}) + E_{\bf y}\,\left[\textrm{OD}({\bf y}) -
  \textrm{FD}({\bf y})\right],
\end{equation}
where the fitted discrepancy $\textrm{FD}({\bf y})$, i.e.,
$d_\textrm{KL}(g_{\bf y};f_{\hat{\bbeta}({\bf y})})$, is simply the
negative log-likelihood $-\ln\mathcal{L}(\hat\bbeta|{\bf y})$ of the fitted
model~$\mathcal{M}_{\hat\bbeta}$.  Since $\textrm{OD}({\bf y})$ is
$d_\textrm{KL}(g;f_{\hat{\bbeta}({\bf y})})$, the second, unbiasing term
in~(\ref{eq:unbiasing}), which was denoted~$B$ in the previous section, can
be calculated from the Gaussian PDF's $g,f_\bbeta$
of~$\mathcal{M}^*,\mathcal{M}_\bbeta$.  They are
\begin{align}
  g({\bf y}) &= (2\pi\sigma_0^2)^{-n/2} \exp \left[-({\bf y}-{\bf y}_0)^t ({\bf y}-{\bf y}_0)/2\sigma_0^2    \right],    \label{eq:gdef}\\
  f_\bbeta({\bf y}) &=  (2\pi\sigma^2)^{-n/2}
\exp \left[-({\bf y}-{\bf X}\bbeta)^t ({\bf y}-{\bf X}\bbeta)/2\sigma^2    \right].\label{eq:fdef}
\end{align}
For convenience we shall write
\begin{equation}
d_\textrm{KL}(f;f)\defeq
d_\textrm{KL}(f_\bbeta,f_\bbeta) = \frac{n}2\left[1 + \ln(2\pi)\right]+ \frac{n}2\ln\left(\sigma^2\right),
\label{eq:firstself}
\end{equation}
since $d_\textrm{KL}(f_\bbeta,f_\bbeta)$ does not depend on~$\bbeta$.  By
examination,
\begin{equation}
\label{eq:FDformula}
\textrm{FD}({\bf y}) =   
-\ln\mathcal{L}_f(\hat\bbeta|{\bf y}) = d_\textrm{KL}(f;f) -\frac{n}2 + \frac{\textrm{RSS}}{2\,\sigma^2}
\end{equation}
expresses the fitted discrepancy in terms of the RSS, which is ${\bf
  y}^t{\bf Q}{\bf y}$.

In the following theorem, $\lambda\defeq {\bf y}_0^t{\bf Q}{\bf
  y}_0/\sigma_0^2$ is an $\mathcal{M}$\nobreakdash-specific
mis-specification parameter, $\chi_r^2$~is a chi-squared random variable
with $r$~degrees of freedom and $\chi_r^2(\lambda)$~is a similar but
non-central variable, with non-centrality parameter~$\lambda$.  (For
central and non-central chi-squared distributions, see Appendix~A\null.)

\begin{mytheorem}
\label{thm:31}
  Under the model\/ $\mathcal{M}^*$, the overall and fitted discrepancies
  of the fitted model\/ $\mathcal{M}_{\hat\bbeta({\bf y})}$, ${\rm OD}={\rm OD}({\bf
    y})$ and\/ ${\rm FD}={\rm FD}({\bf y})$, are distributed according to
  \begin{alignat*}{2}
    {\rm OD}&=d_\textrm{KL}(g;f_{\hat{\bbeta}({\bf y})})&&\sim\, d_\textrm{KL}(f;f) + \frac{n}2 \left(\frac{\sigma_0^2}{\sigma^2}-1\right) + \frac{\sigma_0^2}{2\,\sigma^2}\, \chi_k^2(\lambda),\\
    {\rm FD}&=d_\textrm{KL}(g_{\bf y};f_{\hat{\bbeta}({\bf y})})&& \sim\, d_\textrm{KL}(f;f) - \frac{n}2  + \frac{\sigma_0^2}{2\,\sigma^2}\, \chi_{n-k}^2(\lambda) ,
  \end{alignat*}
where\/ $\hat{\bbeta}({\bf y})=({\bf X}^t{\bf X})^{-1}{\bf X}^t\,{\bf
  y}\in\mathbb{R}^k$ is the fitted value of the parameter\/ $\bbeta$.  For
the discrepancies due to approximation and estimation, ${\rm AD}$ and\/
${\rm ED}={\rm ED}({\bf y})$, the corresponding statements are
  \begin{alignat*}{2}
    {\rm AD}&=d_\textrm{KL}(g;f_{\bbeta_*}) &&=\,  d_\textrm{KL}(f;f) + \frac{n}2 \left(\frac{\sigma_0^2}{\sigma^2}-1\right) + \frac{\sigma_0^2}{2\,\sigma^2}\, \lambda,\\
    {\rm ED}&=d_\textrm{KL}(f_{\bbeta_*};f_{\hat\bbeta({\bf y})})&&\sim\, d_\textrm{KL}(f;f) + \frac{\sigma_0^2}{2\,\sigma^2} \,\chi_k^2 ,
  \end{alignat*}
where\/ ${\bbeta}_*=({\bf X}^t{\bf X})^{-1}{\bf X}^t\,{\bf
  y}_0\in\mathbb{R}^k$ is the pseudo-true value of the parameter\/
$\bbeta$.
\end{mytheorem}
\begin{proof}
  Use the definition~(\ref{eq:dKL}) of~$d_\textrm{KL}$ and the definitions
  (\ref{eq:gdef}),(\ref{eq:fdef}) of~$g,f_\bbeta$.  Each integral
  over~$\mathbb{R}^n$ in the computation of a~$d_\textrm{KL}$ is a normal
  expectation that can be evaluated in closed form.  In the expressions for
  OD,FD,ED the distributions
  $\chi_k^2(\lambda),\chi_{n-k}^2(\lambda),\allowbreak \chi_k^2$ arise
  respectively as the distributions of
  \begin{subequations}
  \begin{align}
    ({\bf P}{\bf y} - {\bf y}_0)^t({\bf P}{\bf y} - {\bf y}_0)/\sigma_0^2 &=({\bf P}{\bf z}_0 - {\bf Q}{\bf y}_0/\sigma_0)^t({\bf P}{\bf z}_0 - {\bf Q}{\bf y}_0/\sigma_0),\label{eq:36a}\\
    ({\bf P}{\bf y} - {\bf y})^t({\bf P}{\bf y} - {\bf y})/\sigma_0^2 &=({\bf z}_0 + {\bf y}_0/\sigma_0)^t{\bf Q}({\bf z}_0 + {\bf y}_0/\sigma_0),\label{eq:fd}\\
    ({\bf P}{\bf y} - {\bf P}{\bf y}_0)^t({\bf P}{\bf y} - {\bf P}{\bf y}_0)/\sigma_0^2 &={\bf z}_0^t{\bf P}{\bf z}_0,
  \end{align}
  \end{subequations}
  if one uses the fact that ${\bf P},{\bf Q}$ are $n\times n$ projection
  matrices of ranks $k,n-\nobreak k$.  (For distributions of quadratic
  forms, see Appendix~A\null.)  Similarly, the~`$\lambda$' in the
  expression for~AD is the non-random value of $({\bf P}{\bf y}_0 - {\bf
    y}_0)^t({\bf P}{\bf y}_0 - {\bf y}_0)/\sigma_0^2$.
\end{proof}

\begin{mytheorem}
\label{thm:exps}
The corresponding expectations over\/ $\mathcal{M}^*$-generated data are
  \begin{alignat*}{2}
    E_{\bf y}\,{\rm OD}({\bf y})\,&=\, d_\textrm{KL}(f;f) &{}+ \frac{n}2 \left(\frac{\sigma_0^2}{\sigma^2}-1\right) &+ \frac{\sigma_0^2}{2\,\sigma^2}\,(k+\lambda),\\
    E_{\bf y}\,{\rm FD}({\bf y})\,&=\, d_\textrm{KL}(f;f) &{}+ \frac{n}2 \left(\frac{\sigma_0^2}{\sigma^2}-1\right) &+ \frac{\sigma_0^2}{2\,\sigma^2}\,(-k+\lambda),\\*[\jot]
    \hphantom{E_{\bf y}}\,{\rm AD}\hphantom{({\bf y})}\,&=\, d_\textrm{KL}(f;f) &+ \frac{n}2 \left(\frac{\sigma_0^2}{\sigma^2}-1\right) &{}+ \frac{\sigma_0^2}{2\,\sigma^2}\,\lambda,\\
    E_{\bf y}\,{\rm ED}({\bf y})\,&=\, d_\textrm{KL}(f;f) &&{}+ \frac{\sigma_0^2}{2\,\sigma^2}\,k.
\end{alignat*}
Thus in expectation only, the variable\/ ${\rm OD}-\nobreak
d_\textrm{KL}(f;f)$ is the sum of\/ ${\rm AD}-\nobreak d_\textrm{KL}(f;f)$
and the variable\/ ${\rm ED}-\nobreak d_\textrm{KL}(f;f)$.
\end{mytheorem}
\begin{proof}
Use $E\left[\chi_r^2(\lambda)\right]=r+\lambda$, with $\chi_r^2$ equalling
$\chi_r^2(0)$.
\end{proof}

\begin{mycorollary}
\label{cor:aic}
  In the definition of the\/ {\rm AIC}-type model-selection criterion\/
  {\rm MSC} for model\/ $\mathcal{M}$, according to which\/ ${\rm MSC}/2$
  equals\/ ${\rm FD}({\bf y})$ plus an unbiasing term\/ $B$, the term\/ $B$
  {\rm(}i.e.\ $E_{\bf y}\,\left[{\rm OD}({\bf y}) -\nobreak {\rm FD}({\bf
      y})\right]${\rm)} equals\/ $\frac{\sigma_0^2}{2\sigma^2}$ times\/
  $2k$.  Hence
  \begin{displaymath}
    {\rm MSC} = 2\,d_{\rm KL}(f;f) + {\rm RSS}/\sigma^2 + \frac{\sigma_0^2}{\sigma^2}\,2k.
  \end{displaymath}
\end{mycorollary}
\begin{proof}
  Compute $E_{\bf y}\,\left[{\rm OD}({\bf y}) -\nobreak {\rm FD}({\bf
      y})\right]$ from the theorem, and then use the
  formula~(\ref{eq:FDformula}) for~$\textrm{FD}({\bf y})$.
\end{proof}

Thus with the exception of a constant term equal to $2d_\textrm{KL}(f;f)$,
which does not affect the relative ranking of models, for any normal linear
model~$\mathcal{M}_\bbeta$ fitted to data the AIC-type selection criterion
reduces to the standard AIC given in~(\ref{eq:AIC}): the usual
$\textrm{RSS}/\sigma^2$, penalized by twice the number of parameters.
Provided, that~is, the model incorporates a~$\sigma^2$ equal to the
variance parameter~$\sigma_0^2$ of the true model~$\mathcal{M}^*$.  There
must be no statistical mis-specification: $\mathcal{M}$~must incorporate
error bars of the correct length.  If so, the formula~(\ref{eq:AIC}) is
exact for all~$n$.  There is no sign of any small-$n$ correction term, such
as appears in the AICc.

It must be stressed that \emph{deterministic} mis-specification is allowed
here.  There may be a~non-zero value for the mis-specification parameter
$\lambda={\bf y}_0^t{\bf Q}{\bf y}_0/\sigma_0^2$, indicating that the
constant vector~${\bf y}_0$ in the definition of the true
model~$\mathcal{M}^*$ does not lie in the estimation space, i.e., the
column space of the design matrix~${\bf X}$; so the
model~$\mathcal{M}_\bbeta$ does not agree with the true
model~$\mathcal{M}^*$ for any value of~$\bbeta$.  Because the parameter
$\lambda$~ appears in both $E_{\bf y}\textrm{OD}({\bf y})$ and $E_{\bf
  y}\textrm{FD}({\bf y})$, it cancels.

\smallskip
The preceding analysis was based entirely on the model-selection policy
of~$\S\,\ref{sec:2}$, according to which the expected overall discrepancy
$E_{\bf y}\,\textrm{OD}({\bf y})$ of the true model~$\mathcal{M}^*$ from a
candidate model~$\mathcal{M}$ should be used for selection purposes.  This
policy gives rise to the AIC, but it is interesting to consider the effects
of generalizing~it slightly.  By the formulas of Theorem~\ref{thm:exps},
this policy is equivalent to choosing the model with the smallest value of
the sum $\textrm{AD} +\nobreak E_{\bf y}\,\textrm{ED}({\bf y})$, the
discrepancy due to approximation plus the expected discrepancy due to
estimation.  Suppose that instead, one assessed the goodness of fit of
$\mathcal{M}$ to~$\mathcal{M}^*$ by employing (any multiple~of) the convex
combination
\begin{equation}
\label{eq:newAIC}
  \bigl(\tfrac1{\gamma}\bigr) \textrm{AD} + \bigl(1-\tfrac1{\gamma}\bigr)\,E_{\bf y}\,\textrm{ED}({\bf y}),
\end{equation}
where $\gamma\ge1$ is free.  By increasing~$\gamma$ one emphasizes the
discrepancy due to estimation, rather than the discrepancy
of~$\mathcal{M}^*$ from~$\mathcal{M}$ due to approximation (which if there
were no mis-specification would be a constant, i.e., would effectively be
zero).  If there is no statistical mis-specification
($\sigma^2=\sigma_0^2$), it follows from the formulas in the theorem that
an unbiased estimator of this convex combination, obtained by
unbiasing~$\textrm{FD}$, is $\textrm{MSC}_\gamma/\gamma$, where
\begin{displaymath}
  {\rm MSC}_\gamma = 2\,d_{\rm KL}(f;f) + {\rm RSS}/\sigma^2 +
  \gamma\, k.
\end{displaymath}
With its constant first term dropped, the model-selection criterion
$\textrm{MSC}_\gamma$ becomes what is widely known as the
$\mathrm{AIC}_\gamma$, which penalizes any model by $\gamma$~times its
number of parameters.  As $\gamma$~increases, over-fitting is increasingly
disfavored.

The conceptual difference between the two sorts of error in statistical
model fitting was pointed~out by \myciteasnoun{Inagaki77}, and an
$\textrm{AIC}_\gamma$-like criterion resembling~(\ref{eq:newAIC}) was
defined for autoregressive models by \myciteasnoun[Eq.~(2.12)]{Bhansali86}.
But it seems not have been noticed that for normal linear models,
$\mathrm{AIC}_\gamma$ arises rather naturally.  Of~course in applications,
domain-specific considerations that are less axiomatic than practical may
affect the choice of~$\gamma$.

\subsection{AIC variability and a significance test}
\label{subsec:3b}

The procedure of deciding among candidate normal linear models will now be
placed in the classical hypothesis testing framework.  A test of
significance for the evidence that $\mathcal{M}^1,\mathcal{M}^2$ are not
equally close to the true data-generating process~$\mathcal{M}^*$ will be
proposed.  The test is valid in the $n\to\infty$ limit, when applied to
models that in a certain precise sense, are separately mis-specified.  It
is assumed that the data points are accompanied by error bars, i.e., that
the residual variance $\sigma_0^2$ is known and is incorporated
in~$\mathcal{M}^1,\mathcal{M}^2$.

The proposed test is based on the statistic $\Delta^{12}\defeq\allowbreak
\textrm{AIC}^2 -\nobreak \textrm{AIC}^1$ and an expression for its
variance, and can potentially replace the traditional use of Akaike
weights.  Focusing on the variability and hence the significance of an AIC
difference has much in common with the approaches of \myciteasnoun{Efron84}
and \myciteasnoun{Fraser87}.  But unlike \citeauthor{Efron84} we do not use
a bootstrap procedure, and unlike \citeauthor{Fraser87} we allow arbitrary
mis-specification.  The general approach is distinguished from the
likelihood ratio testing approach originating with \myciteasnoun{Cox62}, in
that it decides between $\mathcal{M}^1,\mathcal{M}^2$ on the basis of which
is \emph{closer to the truth}, not on the basis of which is \emph{more
  likely to be correct}.

For regression applications, consider the case when the models are fitted
to a random vector ${\bf y}\in\mathbb{R}^n$ that is generated by an unknown
true model ${\bf y}={\bf y}_0 +\nobreak \bepsilon_0$, and $N=1$: only one
observation of the random vector~${\bf y}$ is available as a datum.  The
candidates $\mathcal{M}^l$, $l=1,2$, are defined by ${\bf
  y}^{(l)}=\allowbreak {\bf X}^{(l)}\bbeta^{(l)}+\nobreak \bepsilon^{(l)}$,
where ${\bf X}^{(l)}$ is an $n\times\nobreak k_l$ design matrix of full
rank (with $k_l<n$) and $\bbeta^{(l)}\in\mathbb{R}^{k_l}$ is a column
vector of parameters.  The estimation space $L_l\subset\mathbb{R}^n$ (i.e.,
the column space of~${\bf X}^{(l)}$) is a linear subspace of
dimension~$k_l$.  Since the models are given, the subspaces $L_1,L_2$ are
specified in advance; thus the analysis below is in a sense conditional.
The error vectors $\bepsilon,\bepsilon^{(1)},\bepsilon^{(2)}$ are taken to
be $\sigma_0{\bf z}_0,\sigma_0{\bf z}^{(1)},\sigma_0{\bf z}^{(2)}$, where
${\bf z}_0,{\bf z}^{(1)},{\bf z}^{(2)}$ are vectors of $n$~independent
standard normal variables.

In this setting the $d_\textrm{KL}$-based MSC of AIC type reduces to the
standard AIC, by Corollary~\ref{cor:aic}.  Dropping the additive constant
$d(f;f)$, we write
\begin{subequations}
\label{eq:aicquadform}
\begin{align}
\textrm{AIC}^{l} &=\textrm{RSS}^{l}/\sigma_0^2 + \nobreak 2\,k_l,\label{eq:aicquadforma}\\
  \textrm{RSS}^{l}/\sigma_0^2 &= \bigl[{\bf y}-\hat{\bf y}^{(l)}\bigr]^t
  \bigl[{\bf y}-\hat{\bf y}^{(l)}\bigr]/\sigma_0^2 \label{eq:aicquadformb}\\
  &= {\bf y}^t{\bf Q}^{(l)}{\bf y}/\sigma_0^2 
  \,=\,({\bf z}_0+{\bf y}_0/\sigma_0)^t{\bf Q}^{(l)}({\bf z}_0+{\bf y}_0/\sigma_0),\nonumber
\end{align}
\end{subequations}
since $\hat{\bf y}\defeq {\bf P}^{(l)}{\bf y}$.  (Cf.~(\ref{eq:fd})).  The
$n\times n$ matrices ${\bf P}^{(l)},{\bf Q}^{(l)}$ project onto
$L_l,L_l^\perp\subset\nobreak\mathbb{R}^n$, with ${\bf
  Q}^{(l)}=\allowbreak{\bf I}_n-\nobreak {\bf P}^{(l)}$; note that $\tr
{\bf P}^{(l)}=k_l$ and $\tr {\bf Q}^{(l)}=\allowbreak n-\nobreak k_l$.

Being an inhomogeneous quadratic form in~${\bf z}_0$,
$\textrm{RSS}^{l}/\sigma_0^2$ has a non-central chi-squared distribution,
which is $\chi_{n-k_l}^2(\lambda^{(l)})$, where $\lambda^{(l)}\defeq {\bf
  y}_0^t{\bf Q}^{(l)}{\bf y}_0/\sigma_0^2$ characterizes the
mis-specification of model~$\mathcal{M}^l$ against~$\mathcal{M}^*$.
$\textrm{AIC}^{l}$ is therefore distributed by
\begin{equation}
\label{eq:aicdistrib}
  \textrm{AIC}^{l} \sim \chi_{n-k_l}^2(\lambda^{(l)}) + 2\,k_l.
\end{equation}
It should be noted that as $r\to\infty$, the distribution of
$\chi_r^2(\lambda)$ is increasingly normal, whether or not $\lambda$~grows
with~$r$; thus as~$n\to\infty$, the distribution of $\textrm{AIC}^{l}$ is
increasingly normal.  But it is the distribution of
$\Delta^{12}\defeq\allowbreak \textrm{AIC}^2 -\nobreak \textrm{AIC}^1$ that
is of interest in model selection, and this is determined by the joint
distribution of $\textrm{AIC}^1,\textrm{AIC}^2$, and hence by the joint
distribution of the quadratic forms ${\bf y}^t{\bf Q}^{(1)}{\bf y}$ and
${\bf y}^t{\bf Q}^{(2)}{\bf y}$.  As will be seen, obtaining an
$n\to\infty$ limit theorem requires that the relationship between~${\bf
  Q}^{(1)},{\bf Q}^{(2)}$ be somewhat restricted.

\begin{mytheorem}
\label{thm:expsandvars}
  The expectation and variance of\/ ${\rm AIC}^1,{\rm AIC}^2$ and the
  difference\/ $\Delta^{12}\defeq\allowbreak {\rm AIC}^2 -\nobreak {\rm
    AIC}^1$ are given by
  \begin{sizealignat}{\small}{3}
    E\,{\rm AIC}^{l}&=&& \tr{\bf Q}^{(l)} + {\bf y}_0^t{\bf   Q}^{(l)}{\bf y}_0/\sigma_0^2 + 2\,k_l&=&\ (n - k_l) + \lambda^{(l)} + 2\,k_l, \nonumber\\
    \Var {\rm AIC}^{l} &=&& 2\tr{\bf Q}^{(l)} + 4\,{\bf y}_0^t{\bf
      Q}^{(l)}{\bf y}_0/\sigma_0^2 &=&\ 2(n-k_l)+4\,\lambda^{(l)},\nonumber\\*
    E\,\Delta^{12}&=&& \tr({\bf Q}^{(2)}-{\bf Q}^{(1)}) + {\bf y}_0^t({\bf  Q}^{(2)}-{\bf  Q}^{(1)}){\bf y}_0/\sigma_0^2&=&\ (k_2-k_1)+(\lambda^{(2)}-\lambda^{(1)}),\nonumber\\
    \Var\Delta^{12}&=&& 2\tr\bigl[({\bf Q}^{(2)}-{\bf Q}^{(1)})^2\bigr] + 4\,{\bf y}_0^t\bigl[({\bf Q}^{(2)}-{\bf  Q}^{(1)})^2\bigr]{\bf y}_0/\sigma_0^2.&&\nonumber
  \end{sizealignat}
\end{mytheorem}
\begin{proof}
  The first three of these follow immediately from~(\ref{eq:aicdistrib}) by
  using $E\left[\chi_r^2(\lambda)\right]=\allowbreak r+\nobreak\lambda$ and
  $\Var\left[\chi_r^2(\lambda)\right]=\allowbreak 2r+\nobreak 4\lambda$.
  All four follow from~(\ref{eq:aicquadform}) by using the known
  expressions for normal moments (i.e., the moments of the components of
  the normal random vector ${\bf y}\in\mathbb{R}^n)$.
\end{proof}

The joint distribution of a pair of quadratic forms in a normal vector such
as ${\bf y}\in\mathbb{R}^n$ is complicated, and in~general can only be
expressed in~terms of special functions \mycite{Mathai92}.  But some cases
can be treated in closed form.  For instance, ${\bf y}^t{\bf Q}^{(1)}{\bf
  y}$ and ${\bf y}^t{\bf Q}^{(2)}{\bf y}$ are independent if (and only~if)
${{\bf Q}^{(1)}{\bf Q}^{(2)}={\bf0}}$, by the Craig--Sakamoto theorem
\mycite{Ogawa2008}.  A~case more important in applications is the
following.  Suppose that normal linear regression models
$\mathcal{M}^1,\mathcal{M}^2$, such as the pair considered here, satisfy
$\mathcal{M}^2\subset\mathcal{M}^1$.  That~is, $\mathcal{M}^2$~is a reduced
version of the fuller model~$\mathcal{M}^1$, obtained by parametric
restriction.  Then they are nested: their estimation subspaces $L_1,L_2$
are related by $L_2\subset L_1$ and $L_2^\perp\supset L_1^\perp$, so that
${\bf Q}^{(1)}{\bf Q}^{(2)}={\bf Q}^{(1)}$.  If the vector~${\bf y}_0$ in
the true model~$\mathcal{M}^*$ satisfies ${\bf y}_0\in L_2\subset L_1$, so
that neither of $\mathcal{M}^1,\mathcal{M}^2$ is mis-specified and
$\lambda^{(1)}=\allowbreak\lambda^{(2)}=\nobreak0$, then in addition to the
distributional statement $\textrm{RSS}^l/\sigma_0^2\sim \chi^2_{n-k_l}$,
one has $(\textrm{RSS}^2 -\nobreak
\textrm{RSS}^1)/\sigma_0^2\sim\allowbreak \chi_{k_1-k_2}^2$.  If
alternatively ${\bf y}_0\in L_1\setminus L_2$, so that $\mathcal{M}^2$~is
mis-specified but the fuller model~$\mathcal{M}^1$ is not, then
\begin{equation}
\label{eq:chi2appearance}
  \Delta^{12} + 2(k_1-k_2) = 
  (\textrm{RSS}^2 -\nobreak
\textrm{RSS}^1)/\sigma_0^2\sim\allowbreak \chi_{k_1-k_2}^2(\lambda^{(2)}).
\end{equation}
Such situations are familiar from multivariate regression, and lead to
(partial) F\nobreakdash-tests of the significance of linear regressors
\mycite{Mardia79}.  However, we wish also to handle ${\bf Q}^{(1)},{\bf
  Q}^{(2)}$ or equivalently $L_1,L_2$ that are less closely related:
non-nestedness and more general mis-specifications should be allowed.

To motivate the proposed hypothesis test a simple limit theorem will now be
proved, on the distribution of~$\Delta^{12}$ in a case often encountered in
the physical sciences.  This is when the models
$\mathcal{M}^1,\mathcal{M}^2$ are at~least slightly mis-specified relative
to the (unknown, presumably infinite-dimensional) true
model~$\mathcal{M}^*$, in the rather consequential sense that each has a
non-zero mean fitting error per data point.  Hence one expects that in the
$n\to\infty$ limit, the mis-specification parameters $\lambda^{(l)}\defeq
{\bf y}_0^t{\bf Q}^{(l)}{\bf y}_0/\sigma_0^2$ will grow proportionately
to~$n$ (generically, at different rates).  For further discussion of
mis-specification regimes, see~\S\,\ref{sec:4}.

In the theorem a certain trace condition will appear as a hypothesis.  It
is motivated by the following consideration.  From $\tr {\bf
  Q}^{(l)}=\allowbreak n-\nobreak k_l$ it follows that $\tr ({\bf
  Q}^{(2)}-\nobreak{\bf Q}^{(1)})=\allowbreak k_1-\nobreak k_2$.  If there
is nestedness and $L_2\subset L_1$, then ${\bf Q}^{(2)}-\nobreak{\bf
  Q}^{(1)}$ is also a projection, and idempotent; thus for any specified
${m\ge1}$, $\tr \left[({\bf Q}^{(2)}-\nobreak{\bf Q}^{(1)})^m\right]$ is
$O(1)$, i.e.\ it does not grow with~$n$.  It is reasonable to suppose that
this condition will hold if $\mathcal{M}^1,\mathcal{M}^2$, even if
non-nested, are sufficiently similar to justify their being used as
competing models of the same data.  (Note that if the condition holds for
$m=2$ then it holds for all $m\ge2$, by a standard trace norm inequality.)
The condition does not hold in the maximally dissimilar case ${{\bf
    Q}^{(1)}{\bf Q}^{(2)}={\bf0}}$, as $\tr \left[({\bf
    Q}^{(2)}-\nobreak{\bf Q}^{(1)})^2\right]$ then equals $\tr{\bf Q}^{(1)}
+\nobreak \tr{\bf Q}^{(2)}$.

\begin{mydefinition}
  Consider a sequence of triples
  $(\mathcal{M}^1,\mathcal{M}^2,\mathcal{M}^*)$ indexed by~$n$ (including
  sequences of vectors ${\bf y}_0\in\mathbb{R}^n$ and design matrices),
  with a common error variance~$\sigma_0^2$.  The $n\times n$ projections
  ${\bf Q}^{(1)},{\bf Q}^{(2)}$ are defined as usual.  If
  as~${n\to\infty}$, ${\bf y}_0^t\left[({\bf Q}^{(2)}-\nobreak{\bf
      Q}^{(1)})^2\right]{\bf y}_0$ is bounded below by a positive multiple
  of~$n$, while (much more routinely) the mis-specification parameters
  $\lambda^{(l)}= {\bf y}_0^t{\bf Q}^{(l)}{\bf y}_0/\sigma_0^2$ of the two
  models are bounded above by a positive multiple of~$n$, the models are
  said to be \emph{asymptotically separately mis-specified}.
\end{mydefinition}

\begin{remark*}
  Nested models $\mathcal{M}^1,\mathcal{M}^2$ will be asymptotically
  separately mis-specified if
  $\lambda^{(1)},\lambda^{(2)},\left|\lambda^{(1)}-\lambda^{(2)}\right|$
  all grow proportionately to~$n$.
\end{remark*}

\begin{mytheorem}
  In this setting of a sequence of triples\/
  $(\mathcal{M}^1,\mathcal{M}^2,\mathcal{M}^*)$ indexed by\/ $n$, if the
  candidate models are asymptotically separately mis-specified, and also
  satisfy the trace condition that\/ $\tr \left[({\bf Q}^{(2)}-\nobreak{\bf
      Q}^{(1)})^2\right]$ is $O(1)$, then the distribution of\/
  $\Delta^{12}\defeq\allowbreak {\rm AIC}^2 -\nobreak {\rm AIC}^1$, the
  expectation and variance of which are given in
  Theorem\/~{\rm\ref{thm:expsandvars}}, is asymptotically normal as\/
  $n\to\infty$.
\end{mytheorem}

\begin{remark*}
  Under the conditions of this theorem, $\Var\Delta^{12}$ will be bounded
  below by a positive multiple of~$n$.  This is because according to
  Theorem~\ref{thm:expsandvars}, $\Var\Delta^{12}$ equals a combination of
  $\tr\bigl[({\bf Q}^{(2)}-{\bf Q}^{(1)})^2\bigr]$ and ${\bf
    y}_0^t\bigl[({\bf Q}^{(2)}-{\bf Q}^{(1)})^2\bigr]{\bf y}_0$.
\end{remark*}

\begin{proof}
  The second cumulant $c_2\left[\Delta^{12}\right]= \Var\Delta^{12}$ is
  bounded below by a positive multiple of~$n$, as just remarked.  It is
  easily seen that each higher cumulant $c_m\left[\Delta^{12}\right]$,
  $m\ge3$, is~$O(n)$.  For instance, $c_3\left[\Delta^{12}\right]$ equals a
  combination of $\tr\bigl[({\bf Q}^{(2)}-\nobreak{\bf Q}^{(1)})^3\bigr]$
  and ${\bf y}_0^t\bigl[({\bf Q}^{(2)}-\nobreak{\bf Q}^{(1)})^3\bigr]{\bf
    y}_0$, and these are respectively $O(1)$ and~$O(n)$.  Therefore the
  moments of $\left(\Delta^{12} -
  E\,\Delta^{12}\right)/\left[\Var\Delta^{12}\right]^{1/2}$ tend to those
  of $N(0,1)$, because its higher cumulants tend to zero.
\end{proof}

The hypothesis test is suggested by the following.  Recall that up to an
unimportant additive constant, $\textrm{AIC}/2$ is an unbiased estimator of
the expected overall discrepancy $E_{\bf y}\textrm{OD}^l({\bf y})$,
i.e.\ the negative of the expected log-likelihood after fitting, the
expectation being over data generated by~$\mathcal{M}^*$.

\begin{mycorollary}
  In the above setting, under the null hypothesis that\/ $E_{\bf y}{\rm
    OD}^1({\bf y})=E_{\bf y}{\rm OD}^2({\bf y})$ for all\/ $n$, i.e.,
  that\/ $\mathcal{M}^1,\mathcal{M}^2$ are equally discrepant from\/
  $\mathcal{M}^*$ for all\/ $n$, the distribution of
  \begin{displaymath}
    \Delta^{12} \Bigm/ 
\Bigl\{
2\tr\bigl[({\bf Q}^{(2)}-{\bf Q}^{(1)})^2\bigr] + 4\,{\bf y}_0^t\bigl[({\bf Q}^{(2)}-{\bf  Q}^{(1)})^2\bigr]{\bf y}_0/\sigma_0^2
\Bigr\}^{1/2}
  \end{displaymath}
tends to $N(0,1)$ as~$n\to\infty$.
\end{mycorollary}
\begin{proof}
  The denominator is $\left[\Var\Delta^{12}\right]^{1/2}$, as given in
  Theorem~\ref{thm:expsandvars}.
\end{proof}

An unbiased estimator of $\Var\Delta^{12}$ is the quantity
\begin{equation}
\label{eq:varest}
  -2\tr\bigl[({\bf Q}^{(2)}-{\bf Q}^{(1)})^2\bigr] + 4\,{\bf y}^t\bigl[({\bf Q}^{(2)}-{\bf  Q}^{(1)})^2\bigr]{\bf y}/\sigma_0^2
\end{equation}
as follows by evaluating its expectation over~${\bf y}$.  It is not
guaranteed to be positive, but the probability of its being so tends to
unity as~$n\to\infty$, since the second term increasingly dominates the
first.  (A~maximum likelihood estimator could perhaps be used instead, but
even when $\mathcal{M}^1,\mathcal{M}^2$ are nested and $\Delta^{12}$ has
essentially a non-central chi-squared distribution as
in~(\ref{eq:chi2appearance}), MLE is difficult to perform
\mycite{Anderson81}.)  The expression~(\ref{eq:varest}) is the key to the
following test, which can be applied at any fixed~$n$.

\begin{hypothesistest*}
 To test the null hypothesis $H_0$ that $\mathcal{M}^1,\mathcal{M}^2$ are
 equally discrepant from the true model~$\mathcal{M}^*$, against the
 alternative that they are not, calculate what is asymptotically an
 $N(0,1)$ test statistic,
 \begin{displaymath}
   z^{(12)} \defeq
    \Delta^{12} \Bigm/ 
\Bigl\{
-2\tr\bigl[({\bf Q}^{(2)}-{\bf Q}^{(1)})^2\bigr] + 4\,{\bf y}^t\bigl[({\bf Q}^{(2)}-{\bf  Q}^{(1)})^2\bigr]{\bf y}/\sigma_0^2
\Bigr\}^{1/2}.
 \end{displaymath}
If $\left|z^{(12)}\right|>z_0$, where $P(\left|Z\right|>z_0)=\alpha$ for a
standard normal variable~$Z$, the evidence against~$H_0$ is significant at
level~$\alpha$.  Equivalently, the $p$-value associated to~$H_0$ is given
by the formula $p=P(\left|Z\right|>\left|z^{(12)}\right|)$.  To~test
against a one-sided alternative that one model is less divergent than the
other, proceed similarly.
\end{hypothesistest*}

Estimating the variance of the AIC difference by~(\ref{eq:varest}) is what
makes this $z$-test possible.  (The small probability that the estimated
variance may be non-positive should be noted.)  It should be stressed that
the normality of the test statistic is a good approximation only for
large-$n$ models $\mathcal{M}^1,\mathcal{M}^2$ that differ appreciably in
their mis-specification.  In~general one would need to exploit the joint
distribution of the forms ${\bf y}^t{\bf Q}^{(1)}{\bf y}$ and ${\bf
  y}^t{\bf Q}^{(2)}{\bf y}$, which is complicated.

This test of the significance of an AIC difference is modelled after a
$z$-test proposed by \myciteasnoun{Linhart88}.  His test is based on the
large-sample ($N\to\infty$) properties of minimum discrepancy estimators,
and is not restricted to normal regression.  Our test applies when $N=1$,
and is valid in the rather different $n\to\infty$ limit.  However, the need
for an asymptotic mis-specification occurs in his analysis, as in ours.  In
its absence, the test statistic could have a limiting non-central
chi-squared distribution, rather than a normal one
\citep[cf\null.][]{Steiger85}

\smallskip
Throughout this section we have dealt with \emph{linear} regression models.
But it is not difficult to extend the estimation of $\Var\Delta^{12}$, and
hence the proposed hypothesis test, to models $\mathcal{M}^1,\mathcal{M}^2$
that are non-linear.  The following is a sketch.  Suppose that model
$\mathcal{M}^l$, $l=1,2$, is defined by ${\bf y}=\allowbreak{\bf
  y}_0^{(l)}+\nobreak\bepsilon^{(l)}$, where ${\bf y}_0^{(l)} = {\bf
  y}_0^{(l)}(\bbeta^{(l)})$ is a sufficiently smooth function, not
necessarily linear, of the parameter vector
$\bbeta^{(l)}\in\mathbb{R}^{k_l}$.  In the non-linear case the estimation
subspace $L_l\subset\mathbb{R}^n$ is replaced by an estimation submanifold
of dimension~$k_l$, but $\mathcal{M}^l$ can be fitted to any datum ${\bf
  y}\in\mathbb{R}^n$ by non-linear regression \mycite{Bates88}.  There will
be a best-fit choice $\hat\bbeta^{(l)}\in\mathbb{R}^{k_l}$ for the
parameter vector, and a predicted vector $\hat{\bf y}^{(l)} = \hat{\bf
  y}^{(l)}_0 (\hat\bbeta^{(l)})$.  As usual, the residual sum of squares
$\textrm{RSS}^{l}$ equals $\bigl[{\bf y}-\hat{\bf y}^{(l)}\bigr]^t
\bigl[{\bf y}-\hat{\bf y}^{(l)}\bigr]$.

The difference from the linear case is this: $\textrm{RSS}^l$ is no~longer
quadratic in~${\bf y}_0$, as in Eq.~(\ref{eq:aicquadformb}).  But it is
straightforward to derive a power series in~${\bf y}_0$ for
$\textrm{RSS}^l$ from a Taylor expansion of ${\bf y}_0(\bbeta^{(l)})$ about
the point $\bbeta^{(l)}=\hat\bbeta^{(l)}$.  In much the same way, one can
obtain an expansion of $\Var\Delta^{12}$ in powers of~${\bf y}_0$.  From
this one can readily construct an unbiased estimator of $\Var\Delta^{12}$
as a power series in~${\bf y}$, by requiring unbiasedness to each order.
By employing a truncation of this series, which is a generalization of the
quadratic estimator~(\ref{eq:varest}), one can extend the proposed test to
candidate normal regression models that are non-linear.  Thus for
non-linear models as for linear ones, it may be possible to employ a
decision procedure that relies on the variability of the $\Delta^{12}$
statistic, rather than on Akaike weights.

\section{Selection without error bars}
\label{sec:4}

The applicability in model selection of the AIC and AICc statistics will
now be considered, in the case when the linear regression models being
assessed are fitted to a data set without error bars.  This is quite
different from the case when error (i.e.\ residual) variances are known and
have been incorporated in each model.  The calculations below reveal the
need for the AICc correction, but also reveal a serious difficulty when a
candidate model is mis-specified by an unknown amount.  It has long been
known that applying the AIC(c) to a mis-specified model is problematic
\mycite{Sawa78,Reschenhofer99}, but we obtain precise expressions for the
asymptotic ($n\to\infty$) shift in the AICc, coming from the
mis-specification.  Our results are similar to those of
\myciteasnoun{Noda96}, but are more explicit.

As in~\S\,\ref{sec:3}\,$\ref{subsec:3a}$, take each candidate regression
model~$\mathcal{M}^l$ to be normal linear, of the form ${\bf y}={\bf
  X}^{(l)}\bbeta^{(l)} +\nobreak \bepsilon^{(l)}$ where $\bepsilon^{(l)} =
\sigma{\bf z}^{(l)}$, with ${\bf z}^{(l)}$ a column vector of independent
standard normals.  The parameter $\btheta^{(l)}=(\bbeta^{(l)};\sigma^2)$
now includes besides $\bbeta^{(l)}\in\mathbb{R}^{k_l}$ the residual
variance~$\sigma^2$, which must also be fitted.  By MLE, if a single datum
${\bf y}\in\mathbb{R}^n$ is available, then $\hat\sigma^2$~equals
$\textrm{RSS}^l / n$.  That~is, $\hat\sigma^2 =\allowbreak {\bf y}^t{\bf
  Q}^{(l)}{\bf y}/n$, where ${\bf Q}^{(l)}$ projects onto the left null
space of~${\bf X}^{(l)}$ (the error space).  The true model~$\mathcal{M}^*$
is ${\bf y}=\allowbreak{\bf y}_0 +\nobreak \bepsilon_0$ with
$\bepsilon_0=\sigma_0{\bf z}_0$, in which both ${\bf y}_0\in\mathbb{R}^n$
and~$\sigma_0^2$ are unknown.

The deterministic mis-specification of~$\mathcal{M}^l$, if any, is
quantified by the parameter $\lambda^{(l)} = {\bf y}_0^t{\bf Q}^{(l)}{\bf
  y}_0/\sigma_0^2$, which is a measure of the distance in~$\mathbb{R}^n$
between ${\bf y}_0$ and the column space of~${\bf X}^{(l)}$ (the estimation
space).  In many reasonable data gathering and regression procedures,
$n$~can be taken arbitrarily large; so the large-$n$ behavior
of~$\lambda^{(l)}$ merits discussion.

One possibility is that $\lambda^{(l)}/n$ will tend to a limit
as~$n\to\infty$, like $\hat\sigma^2 =\allowbreak {\bf y}^t{\bf Q}^{(l)}{\bf
  y}/n$.  That is, in the limit some fraction of the RSS may be
attributable to fitting errors of non-zero mean, coming from
mis-specification, rather than to the random errors of mean zero and
typical size~$\sigma_0$ that come from stochasticity in the data-generating
process~$\mathcal{M}^*$.  (This possibility was discussed
in~\S\,\ref{sec:3}\,$\ref{subsec:3a}$.)  Another possibility is that
$\lambda^{(l)}$ will grow sublinearly in~$n$ or even tend to a finite
value, for a subtle reason: as $n$~increases, it may be possible to enhance
the regression by improving or expanding the model~$\mathcal{M}^l$, giving
an even better fit to~$\mathcal{M}^*$.  But it must be stressed that in the
present framework, which does not make explicit the possibility of taking
$n$ to infinity or even of varying~$n$, there is no way of distinguishing
the fractional contribution made to~$\textrm{RSS}^l$ by a non-zero
mis-specification~$\lambda^{(l)}$, or of estimating its magnitude.  Of
course, in applications where the components $(y_i)_{i=1}^n$ of the
observed vector~${\bf y}$ are the values of a response variable
corresponding to values $(x_i)_{i=1}^n$ of a explanatory one, one may
sometimes be able to estimate this fraction by examining a residual plot.

As in~\S\,\ref{sec:3}\,$\ref{subsec:3a}$, where there was no need to
estimate~$\sigma^2$, an explicit expression for the AIC-type selection
criterion MSC (see Definition~\ref{def:2}) is readily obtained.  The MSC is
defined so that $\textrm{MSC}/2$ for any candidate
$\mathcal{M}=\mathcal{M}_\btheta$ of the form ${\bf y}={\bf
  X}{\bbeta}+\nobreak\bepsilon$, with $\bbeta\in\mathbb{R}^k$, is an
unbiased estimator of the expected overall discrepancy $E_{\bf
  y}\textrm{OD}({\bf y})$ under~$\mathcal{M}^*$.  This is because according
to the policy of~\S\,\ref{sec:2}, it is the latter that should be used in
model selection.  For the discrepancy~$d_\textrm{KL}$, $\textrm{OD}({\bf
  y})$ equals $d_\textrm{KL}(g; f_{\hat\btheta({\bf y})})$, in which
$\hat\btheta=(\hat\bbeta,\hat\sigma^2)$ is the fitted parameter obtained by
MLE\null.  Here $\hat{\bbeta}({\bf y})=({\bf X}^t{\bf X})^{-1}{\bf
  X}^t\,{\bf y}\in\mathbb{R}^k$ as usual; and now $\hat\sigma^2={\bf
  y}^t{\bf Q}{\bf y}/n$, where ${\bf Q}={\bf I}_n-\nobreak{\bf P}$ and
${\bf P}$~projects onto the column space of~${\bf X}$.  The PDF's
$g,f_\btheta$ of $\mathcal{M}^*,\mathcal{M}_\btheta$ are
\begin{align}
  g({\bf y}) &= (2\pi\sigma_0^2)^{-n/2} \exp \left[-({\bf y}-{\bf y}_0)^t ({\bf y}-{\bf y}_0)/2\sigma_0^2    \right],    \label{eq:gdef2}\\
  f_\btheta({\bf y}) &=  (2\pi\sigma^2)^{-n/2}
\exp \left[-({\bf y}-{\bf X}\bbeta)^t ({\bf y}-{\bf X}\bbeta)/2\sigma^2    \right],\label{eq:fdef2}
\end{align}
and by direct computation,
\begin{equation}
  d_\textrm{KL}(f_\btheta,f_\btheta) = C_n + \frac{n}2\ln\left(\sigma^2\right)
\end{equation}
as in~(\ref{eq:firstself}), where we now write $C_n\defeq
(n/2)\left[1+\ln(2\pi)\right]$.  

What can be calculated from the datum ${\bf y}\in\mathbb{R}^n$ is not
$\textrm{OD}({\bf y})$ but the fitted discrepancy $\textrm{FD}({\bf y})$,
i.e., $d_\textrm{KL}(g_{\bf y};f_{\hat\btheta({\bf y})})$, where $g_{\bf
  y}$~is a 1-point atomic PDF\null.  This is simply the negative
log-likelihood $-\ln{\mathcal{L}}_f(\hat\btheta|{\bf y})$ of the fitted
model~$\mathcal{M}_{\hat\btheta}$.  As
in~\S\,\ref{sec:3}\,$\ref{subsec:3a}$, the MSC is given by
\begin{equation}
\label{eq:newmsc}
  \textrm{MSC}/2 = \textrm{FD}({\bf y}) + B \,\defeq\, \textrm{FD}({\bf y}) + E_{\bf y}\left[\textrm{OD}({\bf y}) -
    \textrm{FD}({\bf y})\right],
\end{equation}
where the `$B$' term performs the unbiasing.  Also much as before
(see~(\ref{eq:FDformula})),
\begin{equation}
\label{eq:FDformula2}
\textrm{FD}({\bf y}) =   
-\ln\mathcal{L}_f(\hat\btheta|{\bf y}) = d_\textrm{KL}(f_{\hat\btheta};f_{\hat\btheta}) -\frac{n}2 + \frac{\textrm{RSS}}{2\,\hat\sigma^2}
\end{equation}
expresses the fitted discrepancy in~terms of the RSS, which is ${\bf
  y}^t{\bf Q}{\bf y}$.  But now, under the true model~$\mathcal{M}^*$ the
fitted variance~$\hat \sigma^2$ as~well as the RSS is a random variable.
Since $\hat\sigma^2$ equals $\textrm{RSS}/n$, (\ref{eq:FDformula2})
simplifies to
\begin{equation}
  \textrm{FD}({\bf y}) = 
-\ln\mathcal{L}_f(\hat\btheta|{\bf y}) =
d_\textrm{KL}(f_{\hat\btheta};f_{\hat\btheta}) = 
C_n + \frac{n}2\ln\left(\hat\sigma^2\right).
\end{equation}
The following is the counterpart of Theorem~\ref{thm:31}.  In the
statement, $\lambda\defeq {\bf y}_0^t{\bf Q}{\bf y}_0/\sigma_0^2$ is the
(presumably unknown) $\mathcal{M}$-specific mis-specification parameter.

\begin{mytheorem}
\label{thm:41}
  Under the model\/ $\mathcal{M}^*$, the overall and fitted discrepancies
  of the fitted model\/ $\mathcal{M}_{\hat\btheta({\bf y})}$, ${\rm OD}={\rm OD}({\bf
    y})$ and ${\rm FD}={\rm FD}({\bf y})$, are distributed according to
  \begin{alignat*}{2}
    {\rm OD}&=d_\textrm{KL}(g;f_{\hat{\btheta}({\bf y})})&&\sim\, C_n + \frac{n}2 \ln \left(\hat\sigma^2\right) + \frac{n}2 \left(\frac{\sigma_0^2}{\hat\sigma^2}-1\right) + \frac{\sigma_0^2}{2\,\hat\sigma^2}\, \chi_k^2(\lambda),\\
    {\rm FD}&=d_\textrm{KL}(g_{\bf y};f_{\hat{\btheta}({\bf y})})&& \sim\,
    C_n + \frac{n}2 \ln \left(\hat\sigma^2\right),
  \end{alignat*}
where\/ $\hat\sigma^2$ equals\/ $(\sigma_0^2/n)$ times a random variable
with distribution\/ $\chi_{n-k}^2(\lambda)$, and 
$\chi_k^2(\lambda)$ signifies a random variable that is independent of\/
$\hat\sigma^2$.
\end{mytheorem}
\begin{proof}
  That $\hat\sigma^2$ equals $(\sigma_0^2/n)$ times a random variable with
  non-central chi-squared distribution $\chi_{n-k}^2(\lambda)$ follows from
  the representation
  \begin{equation}
    \hat\sigma^2 / (\sigma_0^2/n) = 
    {\bf y}^t{\bf Q}{\bf y}/\sigma_0^2 =({\bf z}_0 + {\bf y}_0/\sigma_0)^t{\bf Q}({\bf z}_0 + {\bf y}_0/\sigma_0).
  \end{equation}
  (Cf.~(\ref{eq:fd}.)  As in the proof of Theorem~\ref{thm:31},
  $\textrm{OD}({\bf y})$ is calculated by using the
  definition~(\ref{eq:dKL}) of~$d_\textrm{KL}$ and the definitions
  (\ref{eq:gdef2}),(\ref{eq:fdef2}) of~$g,f_\btheta$.  The definite
  integral in the definition of~$d_\textrm{KL}$ can be evaluated in closed
  form, and the resulting quadratic form $({\bf P}{\bf y} -\nobreak {\bf
    y}_0)^t({\bf P}{\bf y} -\nobreak {\bf y}_0)/\sigma_0^2$ has
  distribution $\chi_k^2(\lambda)$.  (Cf.~(\ref{eq:36a}).)  That this and
  the quadratic form $\hat\sigma^2= {\bf y}^t{\bf Q}{\bf y}/n$ are
  independent follows from the `if'~part of the Craig--Sakamoto theorem,
  mentioned above, since the projection matrices ${\bf P},{\bf Q}$ are
  complementary: they satisfy ${\bf P}{\bf Q}={\bf 0}$.
\end{proof}

\begin{mytheorem}
\label{thm:42}
The unbiasing term\/ `\,$2B$' in the definition\/ {\rm(\ref{eq:newmsc})} of\/
the\/ {\rm AIC}-type selection criterion\/ {\rm MSC} is expressed in
terms of the moments of non-central chi-squared random variables by
\begin{align*}
  2\,B &= 2\,E_{\bf y}\left[{\rm OD}({\bf y}) - {\rm FD}({\bf y})\right]\\
  &= {n}\,\Bigl\{n\,E\,[(\chi^2_{n-k}(\lambda))^{-1}] - 1 + 
[E\,(\chi^2_k(\lambda))]\,E\,[(\chi^2_{n-k}(\lambda))^{-1}]\Bigr\}\\
  &= n\,\left\{-1 + (n+k+\lambda)\left[
\frac1{n-k-2} - \frac1{(n-k)(n-k-2)}\,\lambda + \dotsb \right]\right\}.
\end{align*}
\end{mytheorem}
\begin{proof}
  This comes from the formulas of Theorem~\ref{thm:41} by exploiting
  $\hat\sigma^2\sim(\sigma_0^2/n)\chi_{n-k}^2(\lambda)$ and independence.
  The first moment $E\,(\chi^2_k(\lambda))$ equals $k+\nobreak\lambda$, and
  the series in~$\lambda$ for the negative first moment
  $E\,[(\chi^2_{r}(\lambda))^{-1}]$ (where $r=n-\nobreak k$), appearing in
  square brackets, is taken from Appendix~A\null.
\end{proof}

The preceding calculation reduces when $\lambda=0$ and $\mathcal{M}$~is not
mis-specified to a derivation of the AICc that has been given by several
authors \mycite{Sugiura78,Hurvich89,Cavanaugh97}.  But it is the
generalization to non-zero~$\lambda$, which is similar to one of
\myciteasnoun{Hurvich91}, which is of~interest.  The chi-squared variables
in the formula for the unbiasing term, which if $\lambda$ were zero would
be central, become non-central with non-centrality parameter~$\lambda$.

As was explained above, in some applications it is reasonable for the
mis-specification~$\lambda$ of a candidate model to be large in the sense
that it grows linearly in~$n$; that~is, if the regression procedure is such
that $n$~can be taken arbitrarily large.  But it is also useful to consider
the case of `medium mis-specification,' when to leading order
$\lambda$~grows proportionately to~$n^{1/2}$, and that of `small
mis-specification,' when $\lambda$~is bounded in~$n$ as~$n\to\infty$.
Hence $\lambda$~will now be allowed to grow according to $\lambda\sim
\lambda_1 n+\lambda_{1/2}n^{1/2} +\nobreak\lambda_0 +\nobreak o(1)$,
$n\to\infty$.

\begin{mytheorem}
\label{thm:43}
Let the additive constant\/ $2C_n$ be dropped from the definition of the\/
{\rm AIC}-type selection criterion\/ {\rm MSC}.  Then if the
mis-specification\/ $\lambda$ of a model\/~$\mathcal{M}$ equals zero, its\/
{\rm MSC} reduces to the standard\/ {\rm AICc} given in\/
{\rm(\ref{eq:AICc}\rm)},
\begin{align*}
{\rm AICc} &= n\ln\left(\hat\sigma^2\right) + \frac{2(k+1)n}{n-k-2}\\
&\sim{\rm AIC} + \frac{2\,(k+1)(k+2)}{n} + O(1/n^2),\qquad n\to\infty,
\end{align*}
where\/ ${\rm AIC}=n\ln\left(\hat\sigma^2\right)+2(k+1)$.  \hfil\break In
the regime of small mis-specification, when\/ $\lambda\sim \lambda_0+o(1)$
as\/ ${n\to\infty}$ with\/ ${\lambda_0>0}$,
\begin{displaymath}
  {\rm MSC}\sim {\rm AICc}  - \frac{\lambda_0(2k+\lambda_0)}n + O(1/n^2), \qquad n\to\infty.
\end{displaymath}
In the regime of medium mis-specification, when\/ $\lambda\sim
\lambda_{1/2}n^{1/2} +o(n^{1/2})$ as\/ ${n\to\infty}$ with\/ ${\lambda_{1/2}>0}$,
\begin{displaymath}
  {\rm MSC}\sim {\rm AIC} - \lambda_{1/2}^2 + o(1/n^{1/2}), \qquad n\to\infty.
\end{displaymath}
In the regime of large mis-specification, when\/ $\lambda\sim \lambda_1 n
+o(n)$ as\/ ${n\to\infty}$ with\/ ${\lambda_1>0}$, the\/ {\rm MSC} equals\/
{\rm AIC} plus a\/ $\lambda_1$-dependent quantity growing with\/~$n$.
\end{mytheorem}
\begin{proof}
  Substitute the leading-order behavior of~$\lambda$ into the formula given
  in Theorem~\ref{thm:42}.  It should be noted that if $\lambda\sim
  \lambda_1 n$, all terms of the power series in~$\lambda$ for
  $E\left[(\chi_{n-k}^2(\lambda))^{-1}\right]$ will contribute, as each
  will be of order~$1/n$.
\end{proof}

This theorem has disconcerting implications for the usefulness of the AIC
and AICc in model selection.  It reveals how different the case of an
unknown error variance $\sigma^2$~is, from the case of a known~$\sigma^2$
(treated in~\S\,\ref{sec:3}\,$\ref{subsec:3a}$).

If $\lambda=0$ and the model is not mis-specified, the theorem confirms
that including the standard AICc correction term of magnitude $O(1/n)$ is
justified.  This term may affect the selection procedure if $n$~is small.
But if $\lambda=\lambda_0+\nobreak O(1/n)$ as~$n\to\infty$ with $\lambda_0$
non-zero, the coefficient of~$1/n$ in the correction term will deviate from
the AICc form.  Since $\lambda_0$~is typically not known, this renders
difficult any small-$n$ correcting of the AIC\null.  In the regime of
medium mis-specification the problem is worse: the $O(1)$ unbiasing term
$2(k+\nobreak1)$ in the AIC itself is shifted by an amount depending
on~$\lambda_{1/2}$.  And in the regime of large mis-specification, in which
applications of model selection may well lie, the AIC is shifted by a
potentially large amount, growing with~$n$.  This shift may swamp the term
$2(k+1)$.

Theorem~\ref{thm:43} indicates that when deciding between candidate models
that have been fitted to a data set without error bars, it may be unwise to
use the AICc or even the AIC, if there is any possibility that the models
are mis-specified relative to the true data-generating
process~$\mathcal{M}^*$, and if the amount of mis-specification is unknown
but is expected to be substantial.  Since in the physical sciences
$\mathcal{M}^*$~is typically infinite-dimensional, candidate models that
are mis-specified, at~least to some extent, are expected to occur quite
widely.

\section{Summary and discussion}
\label{sec:5}

In this paper the AIC and AICc were developed from first principles, to
clarify their ability to assess competing regression models of the sort
common in the physical sciences: ones with normal errors and known error
variances, coming from the error bars of a size-$n$ data set.  The data set
was viewed as providing a single observation ($N=1$) of an $S$-valued
random quantity, the space~$S$ being~$\mathbb{R}^n$.

In~\S\,\ref{sec:2} a model selection policy was formulated, applying to
arbitrary $S$ and arbitrary~$N$.  The Kullback--Leibler
divergence~$d_\textrm{KL}$ was then chosen as the discrepancy functional,
which the policy left unspecified.  The choice of~$d_\textrm{KL}$ ensures
that fitted models are compared on the basis of their fitted
log-likelihoods, suitably unbiased (i.e., penalized).  It was noted that
other measures of the discrepancy between a parametric model and a data set
could be used, such as the popular Hellinger distance.  This option is
worth exploring, since MLE is not robust and may not be the best choice if,
say, the regression is non-linear or non-normal errors are present.  But as
was explained, this will require that the selection policy be modified to
include a form of kernel density estimation.

It was shown in~\S\,\ref{sec:3}\,$\ref{subsec:3a}$ that when fitting a
linear regression model to data with error bars, the AIC and not the AICc
should be used.  (For comments on the recent astrophysics literature, see
Appendix~B\null.)  If the model incorporates a known error
variance~$\sigma^2$, its mis-specification if~any does not affect the
validity of the AIC, though it causes certain discrepancy statistics to
have non-central rather than central chi-squared distributions.  That no
additional unbiasing of the AIC is needed when $\sigma^2$~is known has
in~fact been noticed \mycite[p.~209]{Kuha2004}, but seems to have attracted
little attention.  In applications of the AIC in the physical sciences, it
is of considerable importance.

In~\S\,\ref{sec:3}\,$\ref{subsec:3b}$, it was shown that in the same
setting as that of~\S\,\ref{sec:3}\,$\ref{subsec:3a}$, the variability of
an AIC difference can be estimated.  A~test of significance was proposed,
which exploits what under reasonable conditions of mis-specification is the
asymptotic ($n\to\infty$) normality of this statistic.  The significance
test is a \emph{test for selection}, which can potentially replace the use
of Akaike weights in deciding between regression models with known error
variances.

The approach of~\S\,\ref{sec:3}\,$\ref{subsec:3b}$ to model selection
resembles the approach of \myciteasnoun{Commenges2008}.  In a general
large-sample ($N\to\infty$) context, not focused on the comparison of
regression models, they proposed a test for selection based on a difference
of two AIC's.  They were able to work~out the asymptotic distribution of a
normalized version of this difference by exploiting large-sample theory for
the likelihood ratio statistic.  This included the classical result of
\myciteasnoun{Wald43} on the comparison of nested models, involving a
non-central chi-squared distribution, and an asymptotic normality result of
\myciteasnoun{Vuong89}, who dealt with non-nested models.  The test
proposed in~\S\,\ref{sec:3}\,$\ref{subsec:3b}$ is similar in spirit, but in
the formulation used here the $n\to\infty$ limit of a regression model is
rather different from an $N\to\infty$ large-sample limit, and requires its
own analysis.

In~\S\,\ref{sec:4} the usual derivation of the AICc statistic, applying to
linear regression models fitted to data sets without error bars, was
extended to models with non-zero mis-specification~$\lambda$.  The
appearance of a non-central chi-squared in the distribution of the overall
Kullback--Leibler discrepancy is not unexpected.  But the behavior of the
unbiasing term in the large-$n$ limit is cause for concern.  The case when
the model mis-specification~$\lambda$ is $o(1)$ as~$n\to\infty$ is the
nicest.  (It has a close analogue in the large-sample theory of the
likelihood ratio: the case of `local alternatives,' when the true value of
a model parameter is taken to approach the pseudo-true value
as~$N\to\infty$.)  Except in this case, the shift in the AICc due to the
mis-specification may swamp in the large-$n$ limit the AICc correction
term, and even the usual $2(k+1)$ unbiasing term.  For mis-specified
regression models fitted to data without error bars, this may well affect
the usefulness of the AICc as a tool in model selection.  The extent to
which this problem occurs in the physical sciences remains to be studied.

\appendix

\section*{Appendix~A\null.  Non-central chi-squared distributions and quadratic forms}

A (central) $\chi^2$ distribution with $r$~degrees of freedom,
denoted~$\chi_r^2$, is the distribution of the sum of the squares of
$r$~independent standard normal random variables.  That is, if ${\bf z}$~is
a column vector of $r$~standard normals, then ${\bf z}^t{\bf
  z}\sim\chi_r^2$.  There is a generalization: if ${\bf P}$ is an $n\times
n$ projection matrix of rank~$s$ with $0\le s\le r$, the quadratic form
$({\bf P}{\bf z})^t({\bf P}{\bf z}) = {\bf z}^t{\bf P}{\bf z}$ has
distribution~$\chi_s^2$.

A $\chi^2$ distribution with $r$~degrees of freedom and non-centrality
parameter~$\lambda$, denoted $\chi_r^2(\lambda)$, is the distribution of
$({\bf z}+{\bf u})^t({\bf z}+{\bf u})$, where ${\bf u}$~is a fixed column
vector.  That is, it is the distribution of the squares of $r$~independent
unit-variance normal variables, not necessarily of mean zero.  The
parameter~$\lambda$ equals~${\bf u}^t{\bf u}$.  There is a generalization:
the quadratic form $[{\bf P}{\bf z}+\nobreak{\bf u}]^t[{\bf P}{\bf
    z}+\nobreak{\bf u}] $ has distribution $\chi_s^2(\lambda)$ with
$\lambda={\bf u}^t{\bf u}$.  A second generalization is that $[{\bf P}({\bf
    z}+\nobreak{\bf u})]^t[{\bf P}({\bf z}+\nobreak{\bf u})] =\allowbreak
({\bf z}+\nobreak{\bf u})^t{\bf P}({\bf z}+\nobreak{\bf u})$ has
distribution $\chi^2_s(\lambda)$ with $\lambda={\bf u}^t{\bf P}{\bf u}$.

When $r>1$, the PDF of $\chi_r^2(\lambda)$ cannot be expressed in terms of
elementary functions, though it can in~terms of the confluent
hypergeometric function~${}_0F_1$, or alternatively a modified Bessel
function of the first kind.  If $X\sim \chi_r^2(\lambda)$ then $X$~has mean
and variance
\begin{equation}
  E\,X = r+\lambda, \qquad \Var X = 2\,r + 4\,\lambda,
\end{equation}
and negative first moment
\begin{equation}
  E\,[X^{-1}] = e^{-\lambda/2}\sum\nolimits_{m=0}^\infty \frac{(\lambda/2)^m}{m!} \,
\frac1{r-2+2\,m}.
\end{equation}
For details, see \myciteasnoun{Mathai92} and \myciteasnoun{Bock84}.

\section*{Appendix~B\null.  The AICc in recent papers}

The AIC and AICc have recently entered the physical sciences and in
particular astrophysics by being used to compare cosmological models.  Such
models have a relatively small number of parameters, and competing models
are usually not nested.  Models have been compared, e.g., on the basis of
their predictions of the distance--redshift relation, which characterizes
the expansion of the Universe.  After a model is fitted by non-linear
regression to observational data, its goodness of fit is assessed by
calculating its AIC or~AICc.  One recent comparison of models, employing
the AIC and the Bayesian criterion BIC, is that of \myciteasnoun{Shi2012}.

A search reveals that many though not all publications in this area use the
AIC and AICc in a fashion that on the basis of the present work, can be
considered correct.  If the observational data are accompanied by error
bars, or a common error variance~$\sigma^2$ is known or assumed, the AIC
should be used; and if the variance is treated as a nuisance parameter to
be fitted, the AICc should be used.  \myciteasnoun{Davis2007},
\myciteasnoun{Li2010} and \myciteasnoun{Tan2012} employ the AICc, despite
their data sets being accompanied by error bars, which strictly speaking is
incorrect; but they observe in their analyses that the AICc correction is
of negligible size and does not affect model comparisons.  The paper of
\citeauthor{Tan2012} is especially valuable from a statistician's point of
view, because they investigate AIC(c) variability empirically rather than
theoretically, using a bootstrap procedure.

Unfortunately, a number of papers in the literature are based on data with
explicit error bars, but use the AICc without commenting on whether its
$O(1/n)$ correction term affects their results.  This includes the papers
of \myciteasnoun{Biesiada2009}, \myciteasnoun{February2010},
\myciteasnoun{Kelly2010}, \myciteasnoun{Dantas2011},
\myciteasnoun{Basilakos2012}, \myciteasnoun{Papageorgiou2012} and
\myciteasnoun{Wang2012}.  A re-examination of the model comparisons in
these papers is surely desirable.


\setlength\bibsep{0pt}
\makeatletter
   \@setfontsize\small\@ixpt{10.75}
\makeatother


\end{document}